\documentclass[journal,draftclsnofoot, onecolumn, 12pt]{IEEEtran}
\usepackage[T1]{fontenc}
\ifodd 1
\newcommand{\discuss}[1]{\textbf{\color{magenta} (Discussion: #1)}} % DISCUSSION
\else

\newcommand{\discuss}[1]{}
\fi
\usepackage{amssymb,amsmath,amsfonts,amsthm}
\usepackage[ruled,vlined,linesnumbered]{algorithm2e}
\usepackage[nospace,noadjust]{cite}
\usepackage{float}
\usepackage{graphicx,epstopdf}
\usepackage{subfigure}
\usepackage{xspace}
\usepackage[usenames,dvipsnames]{color}
\usepackage{array}
\usepackage{ctable}
\usepackage{enumerate}
\usepackage{hyperref}
\usepackage{breakurl}
\theoremstyle{plain}

\newtheorem{defn}{Definition}
\newtheorem{prop}{Proposition}
\newtheorem{assum}{Assumption}
\newtheorem{lemma}{Lemma}
\newtheorem{remark}{Remark}

\begin{document}
\title{Energy efficient D2D communications in dynamic TDD systems}

\author{\IEEEauthorblockN{Demia Della Penda, Liqun Fu, and Mikael Johansson}\\
\thanks{All the authors are with the Department of Automatic Control, KTH-Royal Institute of Technology, Sweden (E-mails: \{demiadp, liqun,  mikaelj\}@kth.se).}
}

%\markboth{IEEE JOURNAL ON SELECTED AREAS IN COMMUNICATIONS}{Submitted paper}

\maketitle

\vspace{-1.5cm}
%%%%%% ABSTRACT %%%%%
\begin{abstract}
\vspace{-0.2cm}

Network-assisted device-to-device communication
is a promising technology for improving the performance
of proximity-based services. 
This paper demonstrates how the integration of device-to-device communications and dynamic time-division duplex can improve the energy efficiency of future cellular networks, leading to a greener system operation and a prolonged battery lifetime of mobile devices.
We jointly optimize the mode selection, transmission period and power allocation to minimize the energy consumption (from both a system and a device perspective) while satisfying a certain rate requirement. The radio resource management problems are formulated as mixed-integer nonlinear programming problems.
Although they are known to be NP-hard in general, we exploit  the problem structure to design efficient algorithms that optimally solve several problem cases. For the remaining cases, a heuristic algorithm that computes near-optimal solutions while respecting practical constraints on execution times and signaling overhead is also proposed.
Simulation results confirm that the combination of device-to-device and flexible time-division-duplex technologies can significantly enhance spectrum- and energy-efficiency of next generation cellular systems.
\end{abstract}

\vspace{-0.2cm}
\begin{IEEEkeywords}
\vspace{-0.2cm}
D2D communication, energy-efficient network, dynamic TDD, mode selection, power control. 
\end{IEEEkeywords}

\IEEEpeerreviewmaketitle

%%%%%%%%%%%%%%%%%%%%%%%%%
%%%%%% INTRODUCTION %%%%%
\section{Introduction}
During the last decade, wireless communications have experienced an explosive growth, both in the number of mobile subscribers and in the data traffic demands.
This phenomenon is expected to continue in the future~\cite{cisco2015}, mainly driven by an increasingly rich web content, video- and audio-streaming and file sharing.
Meeting this increasing demand causes the energy consumption of wireless systems to escalate. The growing energy bills of operators, limited battery lifetime of mobile devices and environmental concerns are progressively steering the research community towards the development of energy-efficient wireless communications~\cite{EE2013}.
Efficiency and scalability are therefore becoming the key criteria for the development of the next generation (5G) systems, where Device-to-Device communication (D2D) is recognized as one of its promising technology components~\cite{5G2014, Tehrani2014}.

A natural question in the context of D2D communication is under which condition two users should communicate through a direct link rather than via the BS. This problem is known as the \emph{mode selection problem}. The optimal mode selection depends on the performance measure that we wish to optimize. 
For example, the authors in~\cite{Doppler2010} select the communication mode to maximize user rate while satisfying SINR constraints on active cellular links. 
In~\cite{Hakola2010}, the authors focus on maximizing the power-efficiency of the network. Maximizing the quality-of-service (QoS) for a given transmit power (and the dual problem of minimizing the power for a given QoS) was considered in~\cite{Shalmashi2014}.
 To realize the full potential of D2D communications, the mode selection should be done jointly with the radio resource allocation. 
Several works have investigated this joint problem~\cite{Jung2012, FodorEurasip2012, Belleschi2014}, mainly developing mixed-integer programming models and dealing with NP-hard formulations. Hence, heuristics~\cite{Belleschi2011, Yu2014,Gao2014} and game-theoretic approaches~\cite{Wu2014} have also been explored as more practical alternatives.

Typically, state of the art literature deals with minimizing the transmission power consumption or maximizing the overall system throughput. 
In this paper, we focus on the energy consumption.
Specifically, we analyze the energy reduction that can be obtained by combining D2D communication with dynamic Time-Division Duplex (TDD) technology, which is considered an attractive duplexing method for 5G networks \cite{Lähetkangas2014, venkatasubramanian2014performance}.
In dynamic TDD systems the BS can adjust the time allocated to uplink (UL) and downlink (DL) traffic dynamically.
Traditionally, this feature has been used to compensate for the asymmetry in the UL  and DL  traffic demands. 
In contrast, we use this degree of freedom to minimize the energy consumption of the transmissions, leveraging on the observation
that the energy required for sending a fixed amount of data decreases with the transmission duration (see~\cite{ElGamal2002} and references therein).

To the best of our knowledge, the energy efficiency improvements that can be obtained by integrating D2D and flexible TDD has not been investigated in literature.
Frameworks for D2D enhanced TDD networks are proposed in 
\cite{venkatasubramanian2015_D2D, Sun2014, Li2011}.
However, they do not account for the  mode selection and they mainly focus on the adaptive UL/DL slot allocation to D2D pairs, so to balance the traffic load, coordinate the interference,  improve coverage probability and  sum-rate.

This paper extends our previous study~\cite{PendaICC2015} to a more realistic multi-link scenario. 
%In the multi-link case it is important how frequency channels are allocated to users to manage interference. 
We consider two possible spectrum allocation strategies: one where each D2D pair is assigned its own frequency channel, and one where all D2D connections share the same channel. 
We develop algorithms that minimize the energy consumption from the perspective of both the entire system and of the mobile devices, since these are the most energy-sensitive part of the network.

Our main contributions can be summarized as follows:
\begin{itemize}
\item We formulate the joint mode selection and resource allocation problems as a mixed-integer nonlinear programs (MINLP), which are NP-hard to solve, in general.
\item For the interference-free case, we propose a low-complexity algorithm that exploits problem structure and finds the optimal solution in polynomial (and in some cases linear) time.
\item When considering interference, finding the optimal solution becomes much more challenging. However, we design an algorithm based on branch and bound (B\&B), which computes the optimal solution in a much more efficient way than a naive exhaustive search. 
\item We propose a heuristic algorithm for computing near-optimal solutions while respecting practical constraints in terms of execution times and signalling overhead.
\end{itemize}

We structure the paper as follows.
\S~\ref{Sec: model} describes the system model and states some basic assumptions. \S~\ref{Sec: Problem Formulation} presents the general problem formulation. \S~\ref{Sec: NoInterf} elaborates the optimal solution for the interference-free case, while \S~\ref{Sec: Interf} presents both the B\&B and the heuristic approach for minimizing total device energy when all D2D communications share a single frequency channel. Numerical results are presented and discussed in \S~\ref{Sec: results}. Finally, \S~\ref{Sec: conclusion} concludes the paper.

%%%%%%%%%%%%%%%%%%%%%%%%
%%%% NETWORK MODEL %%%%%
\section{System Model and Assumption} \label{Sec: model}
We consider a single-cell network where infrastructure-assisted device-to-device communication is enabled.
Communication between in-cell users is done in one of two possible modes:
\begin{enumerate}
\item  \textit{Cellular mode}: the transmitter first sends the data to the BS, which then forwards the message to the intended receiver (Fig.~\ref{sub1_a});
\item  \textit{D2D mode}: a dedicated direct link between the transmitter and receiver is set up (Fig.~\ref{sub1_b}).
%, which allows the users to communicate without involving the BS (Fig.~\ref{sub1_b}).
\end{enumerate}

In the cell, we consider a set ${\mathcal L}$ of user pairs that wish to communicate. Each user pair constitutes a logical link that we label by an integer $1, 2, \dots L$. The BS is denoted by $0$ and we refer to the users in pair-$l$ as transmitter-$l$ (Tx-$l$) and receiver-$l$ (Rx-$l$), respectively.

\begin{figure}[ht]
     \begin{center}
        \subfigure[cellular mode]{
             \includegraphics[width=0.2\hsize]{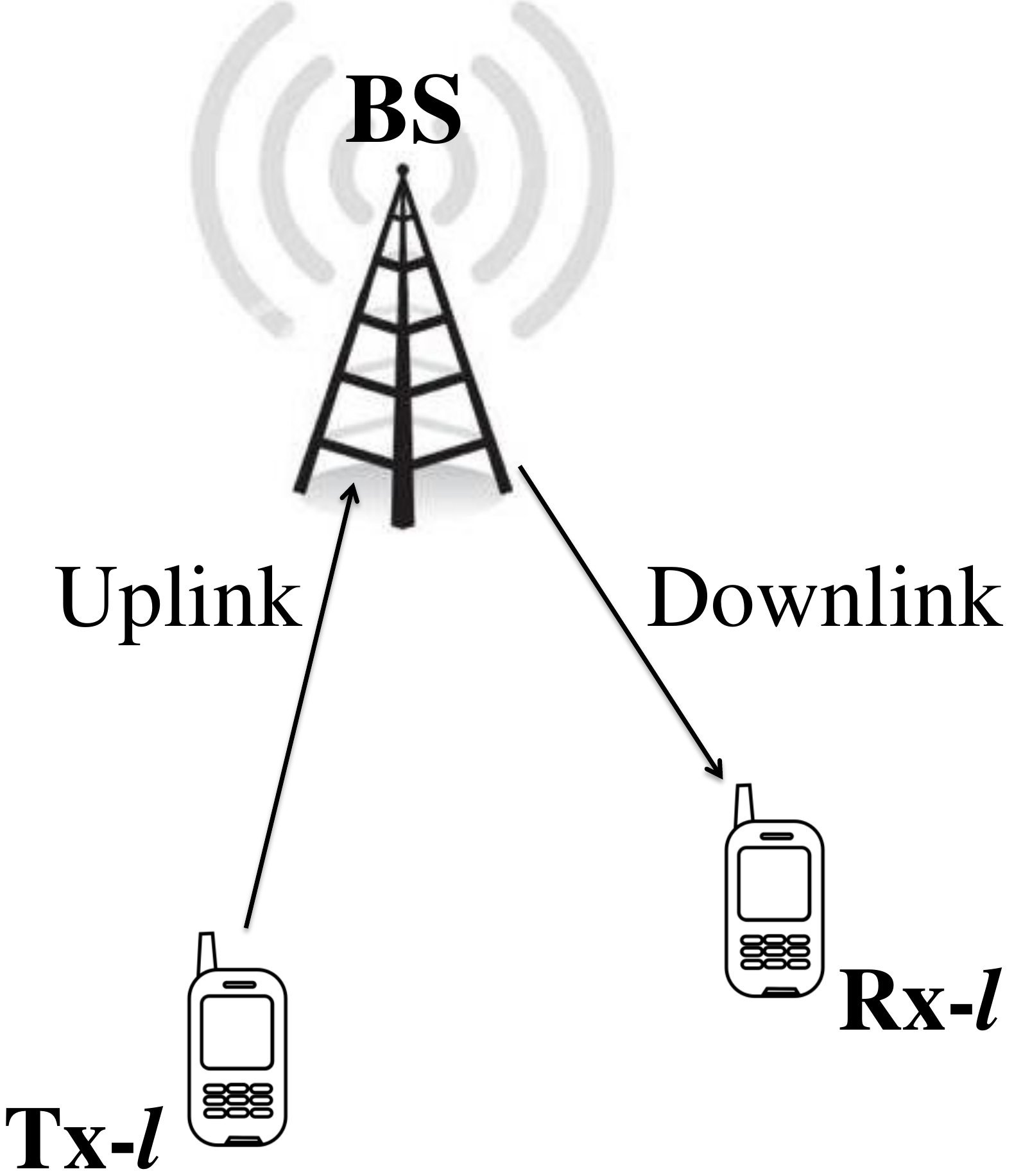}\label{sub1_a}
        }
        \subfigure[D2D mode]{
            \includegraphics[width=0.2\hsize]{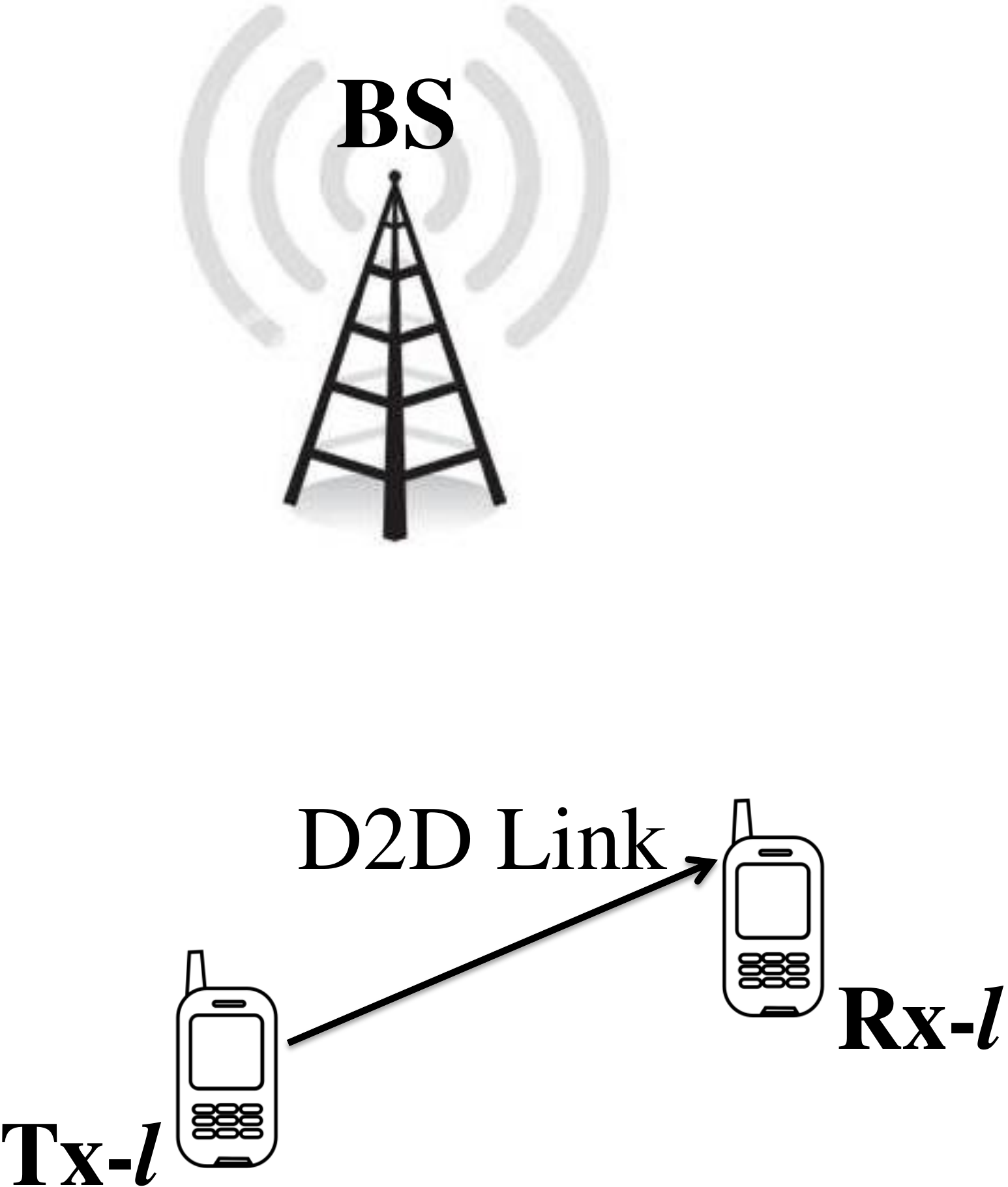}\label{sub1_b}
        }
    \end{center}
\caption{Communication modes: user pair-$l$ can communicate either via the BS (a) or through a direct link (b).}\label{sub1}
\end{figure}

The system bandwidth is divided into orthogonal channels of size $W$ Hz, and time is divided into frames of fixed length of $T$ seconds, see Fig.~\ref{sub2}. The BS manages the spectrum by assigning a time-frequency physical resource block to each logical link.
Each resource block consists of one frame in the time domain, and one channel in the frequency domain.

\begin{figure} [h!]
\begin{center}
\includegraphics[scale=0.35]{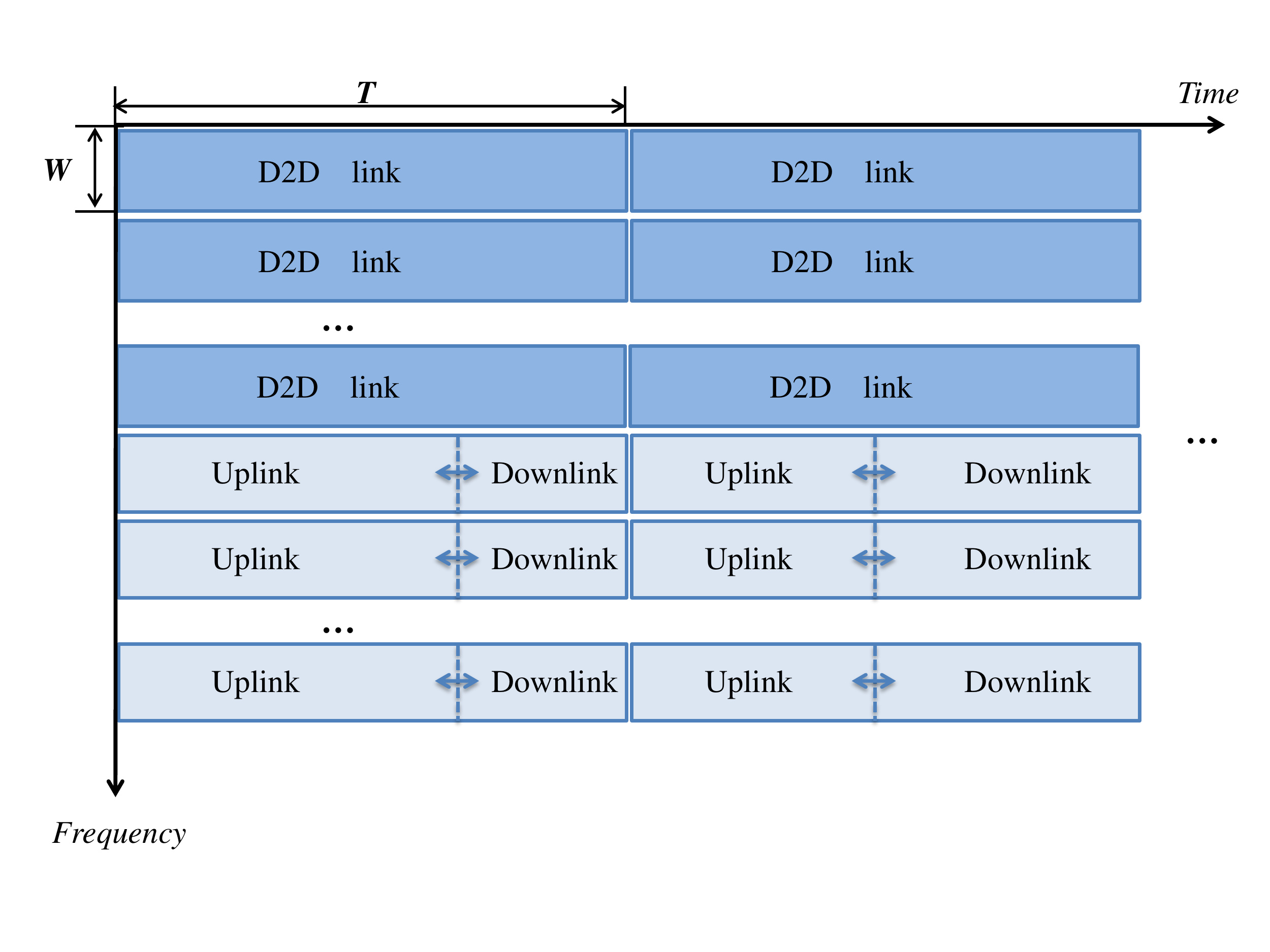}
\end{center}
\caption{Frequency-time resources configuration for D2D communications with dynamic TDD scheme. D2D links can use the full frame duration $T$; UL and DL durations for cellular communications can be reconfigured at each time frame.}\label{sub2}
\end{figure}

We consider \textit{overlay in-band} D2D communication, where D2D communications are allowed to use cellular resources but may not cause interference to traditional cellular communications.
To this end, we assume that the channels are divided into two non-overlapping sets: one allocated to the pairs communicating in cellular mode and one to those communicating in D2D mode.

\subsection{Communication in cellular mode}
To prevent intra-cell interference between concurrent transmissions, the BS follows the channel allocation policy of legacy LTE systems and assigns a separate channel to each user pair in cellular mode.
{We assume that the system adopts a dynamic TDD mode, where the UL and DL transmissions for a user pair occur on the same frequency channel but alternate in time (see Fig.~\ref{sub2}).
In particular, we consider a dynamic TDD system, where the time allocation for UL and DL transmissions can be reconfigured in each time frame, but it is the same for all communications.
This synchronization of cellular communications is a standard practice in current TDD systems to avoid strong inter-cell interference~\cite{Shen2012, TDD_Wang2009}. We denote by $t_{\rm ul}$ and $t_{\rm dl}$ the portion of the time frame allocated to the UL and DL transmissions, respectively.

Let $p_{l0}$ and $p_{0l}$ be the transmit power levels used by Tx-$l$ in the uplink to the BS, and used by the BS in downlink to Rx-$l$, which are subject to bounds $p_{l}^{\rm max}$ and $p_{0}^{\rm max}$, respectively. The instantaneous rates $r_{l0}$ and $r_{0l}$ achieved in UL and DL, respectively, are assumed to follow Shannon's capacity formula:
\begin{equation} \label{eq: rate}
r_{l0} = W \log\left(1 + \frac{p_{l0} G_{l0}}{\sigma^2}\right), \quad r_{0l} = W \log\left(1 + \frac{p_{0l} G_{0l}}{\sigma^2}\right).
\end{equation}
Here, $G_{l0}$ is the channel gain between Tx-$l$ and the BS, $G_{0l}$ is the channel gain between the BS and Rx-$l$, and $\sigma^2$ is the noise power. The maximum instantaneous rates (corresponding to maximum power transmissions) in~\eqref{eq: rate} are denoted by $r_{l0}^{\rm max}$ and $r_{0l}^{\rm max}$, respectively.

\subsection{Communication in D2D mode} \label{IntroD2D}
In D2D mode, each pair can use the full frame duration for its \emph{single-hop} (\cite{Fodor2012}) transmission, as illustrated in Fig.~\ref{sub2}. Let $t_{l}\leq T$ denote the active time of pair-$l$ in D2D mode. As discussed above, the D2D communication occurs on different orthogonal frequency channels than those used for cellular communications.
We consider two  channel allocation strategies for D2D pairs:
\begin{itemize}
\item \textbf{Full Orthogonality (FO)}:
All D2D communications are assigned orthogonal frequency channels. Hence, no receiver is interfered by other transmissions within the cell.
\item \textbf{D2D Resource Sharing (RS)}: All D2D communications are assigned to the same frequency channel, hence interfere with each other.
\end{itemize}

Let $r_{ll}$ denote the instantaneous rate between users of pair-$l$ when transmitting in D2D mode:
\begin{equation} \label{eq: rated2d}
r_{ll}  = W \log\left(1 + \frac{p_{ll} G_{ll}}{\sigma^2 + I_{l}}\right),
\end{equation}
Here, $G_{ll}$ is the direct channel gain between Tx-$l$ and its intended receiver,  $p_{ll}$ is the transmit power level (upper bounded by $p_{l}^{\rm max}$), and $I_l=\sum_{j\neq l}p_{jj}G_{jl}$ is the interference power level experienced at Rx-$l$, due to other concurrent D2D transmissions. Note that $I_l=0$ if we consider the FO channel allocation scheme. We use $r_{ll}^{\rm max}$ to denote the maximum instantaneous achievable rate corresponding to $p_{l}^{\rm max}$.

\subsection{Rate constraint, power feasibility and energy cost} \label{IntroEnergy}
To guarantee a certain QoS, each pair-$l$ has a traffic requirement of  $b_l$ nats
during a time frame, irrespectively of the communication mode. This QoS requirement can be translated into a session rate requirement of $ b_l/T$ nats per second.
Specifically, if pair-$l$ is in cellular mode, the transmission times for UL and DL, along with the corresponding instantaneous transmission rates, must satisfy
\begin{equation} \label{eq: rateConstraint}
r_{l0}t_{\rm ul} \geq  b_l \quad \text{and} \quad r_{0l}t_{\rm dl} \geq  b_l,
\end{equation}
Similarly, if pair-$l$ is in D2D mode, $t_{l}$ and $r_{ll}$ must satisfy
\begin{equation} \label{eq: rateConstraint2}
r_{ll}t_{l} \geq  b_l.
\end{equation}

%%%%%%%%%%%%%%%%%% Power feasibility
The limitation on the transmission power levels, together with the session rate requirements above, entail the need to verify under which conditions the communication of a pair can be supported by the network.
To this end, we introduce the concept of \textit{power-feasibility}:

\begin{defn}[Power feasibility]\label{def: powerFeasibility}
We say that user pair-$l$ is \emph{power-feasible} 
\begin{itemize}
\item[(a)] in D2D mode if  $r_{ll}^{\rm max}T \geq b_l$;
\item[(b)] in cellular mode if there exists a time allocation $(t_{\rm ul}, t_{\rm dl})$ such that
\begin{align}
\begin{cases}
t_{\rm ul} + t_{\rm dl} &\leq T, \\
t_{\rm ul} &\geq \frac{b_l}{r_{l0}^{\rm max}}, \\
t_{\rm dl} &\geq \frac{b_l}{r_{0l}^{\rm max}}.
\end{cases} \label{eqn:cellular_feasibility}
\end{align}
\end{itemize}
\end{defn}

\begin{assum}\label{Ass1}
Every user pair-$l$ in the cell can be supported at least in cellular mode.
\end{assum}

\begin{remark}
The power-feasibility condition (\ref{eqn:cellular_feasibility}) implies that $t_{\rm ul}$ must satisfy
\begin{equation} \label{PowerFeasible}
\frac{b_l}{r_{l0}^{\rm max}} \leq t_{\rm ul} \leq T -\frac{b_l}{r_{0l}^{\rm max}}.
\end{equation}
Thus, Assumption~\ref{def: powerFeasibility} of a feasible UL time allocation for all pairs suggests the following feasibility condition for the session rate requirements:
$\underset{l}{\max} \{ \frac{b_l}{r_{l0}^{\rm max}}\}\leq \underset{l}{\min} \{T -\frac{b_l}{r_{0l}^{\rm max}}\}$.
\end{remark}
%%%%%%%%%%%%%%%%%%%%

By inverting the power-rate relationships~\eqref{eq: rate} and~\eqref{eq: rated2d}, we find the energy required to satisfy the session rate requirement $b_l$ on link-$l$ for a given time allocation
\begin{align}
E_{l0}(t_{\rm ul}) &= p_{l0}t_{\rm ul} = \left( \exp\left(\frac{b_l}{Wt_{\rm ul}}\right) -1 \right)  \frac{\sigma^2}{G_{l0}}t_{\rm ul}, && \text{UL, cellular mode}; \nonumber\\
E_{0l}(t_{\rm dl}) &= p_{0l}t_{\rm dl} = \left( \exp\left(\frac{b_l}{Wt_{\rm dl}}\right) -1 \right)  \frac{\sigma^2}{G_{0l}}t_{\rm dl}, && \text{DL, cellular mode}; \label{eq: energy}\\
E_{ll}^{\rm D2D}(t_l, I_l) &= p_{ll}t_{l} = \left( \exp\left(\frac{b_l}{Wt_{l}}\right) -1 \right)  \frac{\sigma^2 + I_l}{G_{ll}}t_{l}, && \text{D2D mode} \nonumber.
\end{align}
These functions are convex and monotonically decreasing~\cite{ElGamal2002} in their arguments. This observation leads to the following result:
\begin{lemma} \label{lemma: Energy cost D2D}
Any energy-optimal solution must allocate the full frame duration for communication. For the D2D mode, this implies that
\begin{align*}
	\min_{t_l, I_l} E_{ll}^{\rm D2D}(t_l, I_l) = \min_{I_l} E_{ll}^{\rm D2D}(T, I_l),
\end{align*}
while for cellular communication it must hold that
\begin{align*}
	t_{\rm ul}+t_{\rm dl} &= T.
\end{align*}
\end{lemma}

The energy cost for communicating in cellular mode includes the energy cost of both the transmitting device and that of the BS. However, the BS often has access to cheap and abundant energy in comparison with the user equipment, in which case it is relevant to only focus on the device energy. To this end, we consider the following two definitions of energy consumption for a generic user pair-$l$ in cellular mode:
\begin{itemize}
\item The \textbf{System Energy consumption (\textbf{SE})}:
is the energy consumed by both Tx-$l$ in UL and the BS in DL. By Lemma~\ref{lemma: Energy cost D2D}, the total energy cost is obtained by minimizing
\begin{equation} \label{SE-MS}
E_{ll}^{\rm CELL}(t_{\rm ul}) =  E_{l0}(t_{\rm ul}) + E_{0l}(T-t_{\rm ul}),
\end{equation}
which is a convex function of $t_{\rm ul}$.

\item The \textbf{User Energy consumption (\textbf{UE})}:
is the energy consumed by Tx-$l$ in UL transmission, disregarding the energy spent by the BS, that is
\begin{equation} \label{UE-MS}
E_{ll}^{\rm CELL}(t_{\rm ul}) =  E_{l0}(t_{\rm ul}),
\end{equation}
which is a convex and monotonically decreasing function in $t_{\rm ul}$.
\end{itemize}

%%%%%%%%%%%%%%%%%%%%%%%%%%%%%%%
%%%%% PROBLEM FORMULATION %%%%%
\section{Problem Statement} \label{Sec: Problem Formulation}
We consider the problem of jointly optimizing the communication mode, the UL/DL transmission times and the powers allocated to all transmitters, to minimize the energy consumption of the system.
For communications in D2D mode, we consider the two possible spectrum sharing strategies introduced in \S~\ref{IntroD2D}: FO and RS. The energy cost for cellular communications is either the user energy (UE), or the total system energy (SE) presented in \S~\ref{IntroEnergy}. Hence, we consider the following four variations of the energy minimization problem:
\begin{itemize}
\item Fully orthogonal channels - system energy minimization (\textbf{FO-SE}),
\item Fully orthogonal channels - user energy minimization (\textbf{FO-UE}),
\item D2D resource sharing - system energy minimization (\textbf{RS-SE}),
\item D2D resource sharing - user energy minimization (\textbf{RS-UE}).
\end{itemize}
In each case, we develop algorithms that find the optimal mode selection and resource allocation that minimize the chosen energy cost. Specifically, the mode selection policy divides the set of user pairs ${\mathcal L}$ into two subsets: ${\mathcal D}$, representing the pairs that should communicate in D2D mode, and ${\mathcal C}={\mathcal L\setminus \mathcal D}$, representing the user pairs that should communicate in cellular mode. The mode selection is jointly optimized with the allocation of the optimal transmission power to each transmitter, and the optimal UL/DL time allocation for all user pairs in $\mathcal{C}$.

As will be seen in the subsequent sections, the optimal joint mode selection and resource allocation problem for each of these four cases can be formulated as a nonlinear mixed-integer mathematical program.
In  \S~\ref{Sec: NoInterf}, we show that optimal solution for FO-SE and FO-UE can be found in polynomial time, even though the overall problem is not convex.
In \S~\ref{Sec: Interf}, on the other hand, we will show that finding the optimal solution in the RS-SE and RS-UE cases is challenging, mainly due to the large amount of runtime needed and knowledge of all cross-link gains required. Thus, we further propose a heuristic algorithm that efficiently computes near-optimal solutions and takes into account practical implementation aspects.

%%%%%%%%%%%%%%%%%%%%%%%%%%%
%%%%% NO INTERFERENCE %%%%%
\section{Minimum-Energy Mode Selection with Full Orthogonality}  \label{Sec: NoInterf}
In this section, we show how the jointly optimal mode selection and resource allocation with full orthogonality can be found in polynomial time. For ease of exposition, we first derive the optimal solution for a single user pair, and then extend the results to multiple user pairs.

%%%%%%%%%%%%%%%%%%%%%%%%%
\subsection{Single user pair}
We first characterize the minimum energy cost for cellular and D2D communication. 
We show that the optimal system energy (SE) cost can be found by solving a simple convex optimization problem, while the minimal user energy (UE) cost admits an explicit expression. 
\subsubsection{Minimum energy cost for communication in cellular mode}
In cellular mode, the UL/DL time allocation is chosen to minimize one of the following two objectives:

\begin{itemize}

\item Minimizing SE:
The minimum amount of energy of pair-$l$ in cellular mode can be determined by solving the following single-variable convex optimization problem:
\begin{subequations} \label{PB_ntwrk}
\begin{align}
    \underset{ t_{\rm ul}}{\text{minimize}}
        & \quad  E_{l0}(t_{\rm ul}) + E_{0l}(T - t_{\rm ul}) \\
    \text{subject to}
        & \quad  \frac{b_l}{r_{l0}^{\max}} \leq t_{\rm ul}  \leq T -  \frac{b_l}{r_{0l}^{\max}},\label{PB_ntwrk: C1}
\end{align}
\end{subequations}
where constraint~\eqref{PB_ntwrk: C1} ensures power  feasibility in the sense of Definition~\ref{def: powerFeasibility}.
Problem~\eqref{PB_ntwrk} can be solved efficiently using a wide variety of methods, including simple bisection search~\cite{Boyd2004}.
Let $t_{\rm ul}^\star$ denote the optimal solution to~\eqref{PB_ntwrk}. The energy cost of pair-$l$ in cellular mode is then
\begin{equation} \label{E_network}
E_{ll}^{\rm CELL}(t_{\rm ul}^\star) =  E_{l0}(t_{\rm ul}^\star) + E_{0l}(T - t_{\rm ul}^\star).
\end{equation}

\item Minimizing UE:
Here, the only difference from the problem formulation in~\eqref{PB_ntwrk} is that the objective function reduces to $E_{l0}(t_{\rm ul})$. By monotonicity of the objective function, the optimal solution is attained by $t_{\rm ul}^\star = T-\frac{b_l}{r_{0l}^{\max}} $, with the corresponding optimal energy cost
\begin{equation}\label{E_user}
E_{ll}^{\rm CELL}(t_{\rm ul}^\star)  =  E_{l0}(T-\frac{b_l}{r_{0l}^{\max}}).
\end{equation}
\end{itemize}

%%%%%%%%%%%%%%%%%%%%%%%%%
\subsubsection{Minimum energy for communication in D2D mode}
In the D2D mode, no traffic is forwarded through the BS and only the user equipment consumes energy for the connection. The minimum energy cost follows from Lemma~\ref{lemma: Energy cost D2D}, with $I_l=0$:
\begin{equation} \label{E_D2D_nonextended}
E_{ll}^{\rm D2D} =  \left( \exp\left(\frac{b_l}{WT}\right) -1 \right)  \frac{\sigma^2}{G_{ll}}T.
\end{equation}
Equation \eqref{E_D2D_nonextended} is only valid when the D2D mode is power feasible. Since we have ensured power feasibility only for communications in cellular mode, we need to verify that $r_{ll}^{\max}T\leq b_l$ before applying~\eqref{E_D2D_nonextended}. It turns out to be more convenient to work with the extended value function
\begin{equation} \label{E_D2D}
\bar{E}^{\rm D2D}_{ll} = \left\{
  \begin{array}{l l}
    \left( \exp\left(\frac{b_l}{WT}\right) -1 \right)  \frac{\sigma^2}{G_{ll}}T & \quad \text{ if } \, r_{ll}^{\max}T\leq b_l\\
   +\infty & \quad \text{otherwise.}
  \end{array} \right.
\end{equation}
%%%%%%%%%%%%%%%%%% %%%%%%%%%%%%%%%%%% %%%%%%%%%%%%%%%%%% %%%%%%%%%%%%%%%%%% %%%%%%%%%%%%%%%%%%
\subsubsection{Optimal mode selection policy and resource allocation}

The optimal mode selection policy consists in first solving the convex optimization problem to estimate the energy cost for cellular mode, and then comparing it with the energy cost for D2D mode. In the FO-SE scenario, we compare \eqref{E_D2D} with \eqref{E_network}, and in the FO-UE scenario we compare \eqref{E_D2D} with \eqref{E_user}.
The optimal communication mode is simply the one that requires the least amount of energy. 

Once the optimal communication mode and transmission time has been selected, the corresponding optimal powers are easily derived as
\begin{align}
p_{l0}   &= \left( \exp\left(\frac{b_l}{Wt_{\rm ul}^\star}\right) -1 \right)  \frac{\sigma^2}{G_{l0}}, && \text{UL, cellular mode}; \nonumber\\
p_{0l}   &= \left( \exp\left(\frac{b_l}{W(T-t_{\rm ul}^\star)}\right) -1 \right)  \frac{\sigma^2}{G_{0l}}, && \text{DL, cellular mode}; \label{eq: opt powers}\\
{p_{ll}} &= \left( \exp\left(\frac{b_l}{WT}\right) -1 \right)  \frac{\sigma^2}{G_{ll}}, && \text{D2D mode} \nonumber.
\end{align}

It is possible to interpret the mode selection policy in terms of the channel gain (and, so, the physical distance) between the two communicating devices. We will explore this geometrical interpretation in \S~\ref{Sec: results} to characterize regions in the cell where D2D communication is preferable.

%%%%%%%%%%%%%%%%%%%%%%%%%
\subsection{Multiple user pairs }
Under full orthogonality, the main challenge with multiple user pairs is that all communications in cellular mode must use one common UL/DL time allocation. In other words, the optimal UL/DL time allocation must account for the energy consumption of all users.
Under full orthogonality, D2D links do not interfere with each
  other. Thus, the minimum energy cost for each pair in
  D2D mode is the same constant value~\eqref{E_D2D} as in the single user pair case.
The energy consumption in cellular mode, on the other hand,
depends on the UL time allocation.
Let $\mathbf{m}\in \{0,1\}^L$ denote the \textit{mode selection vector} whose entries satisfy
\begin{equation} \label{m}
 m_l = \left\{
  \begin{array}{l l}
    0 & \quad \text{if pair-$l$ is in cellular mode},\\
    1 & \quad \text{if pair-$l$ is in
     D2D mode.}
  \end{array} \right.
\end{equation}
The general formulation of the joint mode selection and time allocation problem for multiple pairs under both FO-SE and FO-UE cases, is the following MINLP:
\begin{subequations} \label{PB_ML}
\begin{align}
   \underset{\mathbf{m}, t_{\rm ul} }{\text{minimize}}
        & \quad  \sum_{l=1}^L  \bar{E}^{\rm D2D}_{ll} m_l + E_{ll}^{\rm CELL}(t_{\rm ul}) (1-m_l)\\
    \text{subject to}
        &\quad  \frac{b_l}{r_{l0}^{\max}} - T m_l \leq t_{\rm ul} \leq T -\frac{b_l}{r_{0l}^{\max}} + T m_l, & \forall l, \label{C: PB_ML} \\
        &\quad t_{\rm ul}\in[0,T], \quad m_l\in \{0,1\}, & \forall l.
\end{align}
\end{subequations}
The objective function is the total energy consumption of all
  the $L$ user pairs, with $E_{ll}^{\rm CELL}(t_{\rm ul})$ given by \eqref{SE-MS} or \eqref{UE-MS} under FO-SE or FO-UE, respectively. 
  Constraints~\eqref{C: PB_ML} ensure that pair-$l$ is assigned to cellular mode, only if it is power feasible in the sense of Definition~\ref{def: powerFeasibility}.
Since the user pairs in D2D mode do not interfere with each other, it turns out to be convenient to eliminate the integer variables by incorporating the optimal and feasible mode selection solution in the single pair case. Thus, problem \eqref{PB_ML} can be recast as
\begin{equation} \label{PB_ML_2}
\begin{split}
    \underset{ t_{\rm ul} \in [0, T]}{\text{minimize}}
         & \quad  F(t_{\rm ul}).
\end{split}
\end{equation}
Here $ F(t_{\rm ul}) = \sum_{l=1}^L E_l(t_{\rm ul})$, and $E_l(t_{\rm ul})$ denotes the minimum energy-cost for the single pair-$l$ when the UL time is fixed to $t_{\rm ul}$, that is
\begin{equation} \label{E_l}
E_l(t_{\rm ul}) \triangleq \begin{cases}
\min \{\bar{E}^{\rm D2D}_{ll}, E_{ll}^{\rm CELL}(t_{\rm ul})\}  & \text{if } t_{\rm ul} \in [\frac{b_l}{r_{l0}^{\max}}, T -\frac{b_l}{r_{0l}^{\max}} ] \\
\bar{E}^{\rm D2D}_{ll}									  & \text{otherwise}.
\end{cases}
\end{equation}

Equation~\eqref{E_l} reveals the piecewise nature of $E_l(t_{\rm{ul}})$. For $t_{\rm{ul}}< b_l/r_{l0}^{\max}$ and  $t_{\rm{ul}}>T -\frac{b_l}{r_{0l}^{\max}}$,  we have $E_l(t_{\rm{ul}}) = \bar{E}_{ll}^{\rm D2D}$, which is a finite constant if user pair-$l$ is power feasible in D2D mode and $+\infty$ otherwise.
In the interval $[b_l/r_{l0}^{\max}, T-b_l/r_{0l}^{\max}]$, $E_l(t_{\rm{ul}})$ is equal to the constant $\bar{E}_{ll}^{D2D}$ or given by the function $E_{ll}^{\rm CELL}(t_{\rm{ul}})$, depending on whether and at which points the graphs of  $E_{ll}^{\rm CELL}(t_{\rm{ul}})$ and $\bar{E}_{ll}^{D2D}$ intersect. Note that the two graphs can intersect only once if $E_{ll}^{\rm CELL}(t_{\rm{ul}})$ is monotonically decreasing (UE) or twice if it is convex (SE). See Fig.~\ref{fig: Example} for an illustration.

To better describe the piecewise nature of $E_{l}(t_{\rm{ul}})$, we introduce $\Delta_l=[\tau_l^{\rm min},\tau_l^{\rm max}]$ as the interval of $t_{\rm ul} $ during which $E_l(t_{\rm ul}) = E_{ll}^{\rm CELL}(t_{\rm ul}) $. 
If, for a pair-$l$, such an interval does not exist (\emph{i.e.}, $\Delta_l=\emptyset$), D2D mode is always more energy efficient than cellular mode for that user pair.

\begin{figure}[h]
     \begin{center}
        \subfigure[FO-UE single link.]{
        \label{fig: Example_1}
            \includegraphics[scale=0.33]{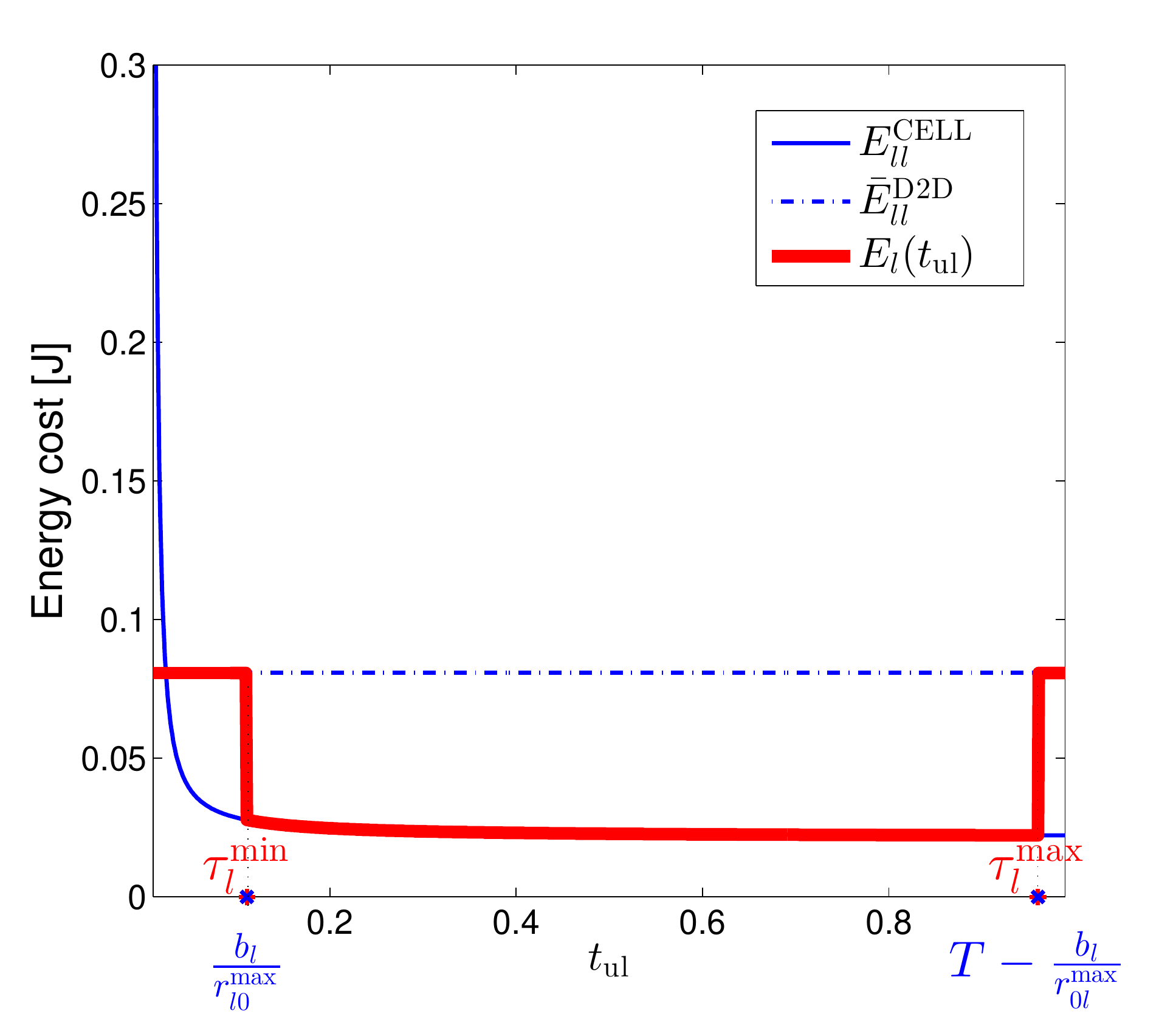}
        }
        \subfigure[FO-SE single link.]{
        \label{fig: Example_2}
           \includegraphics[scale=0.33]{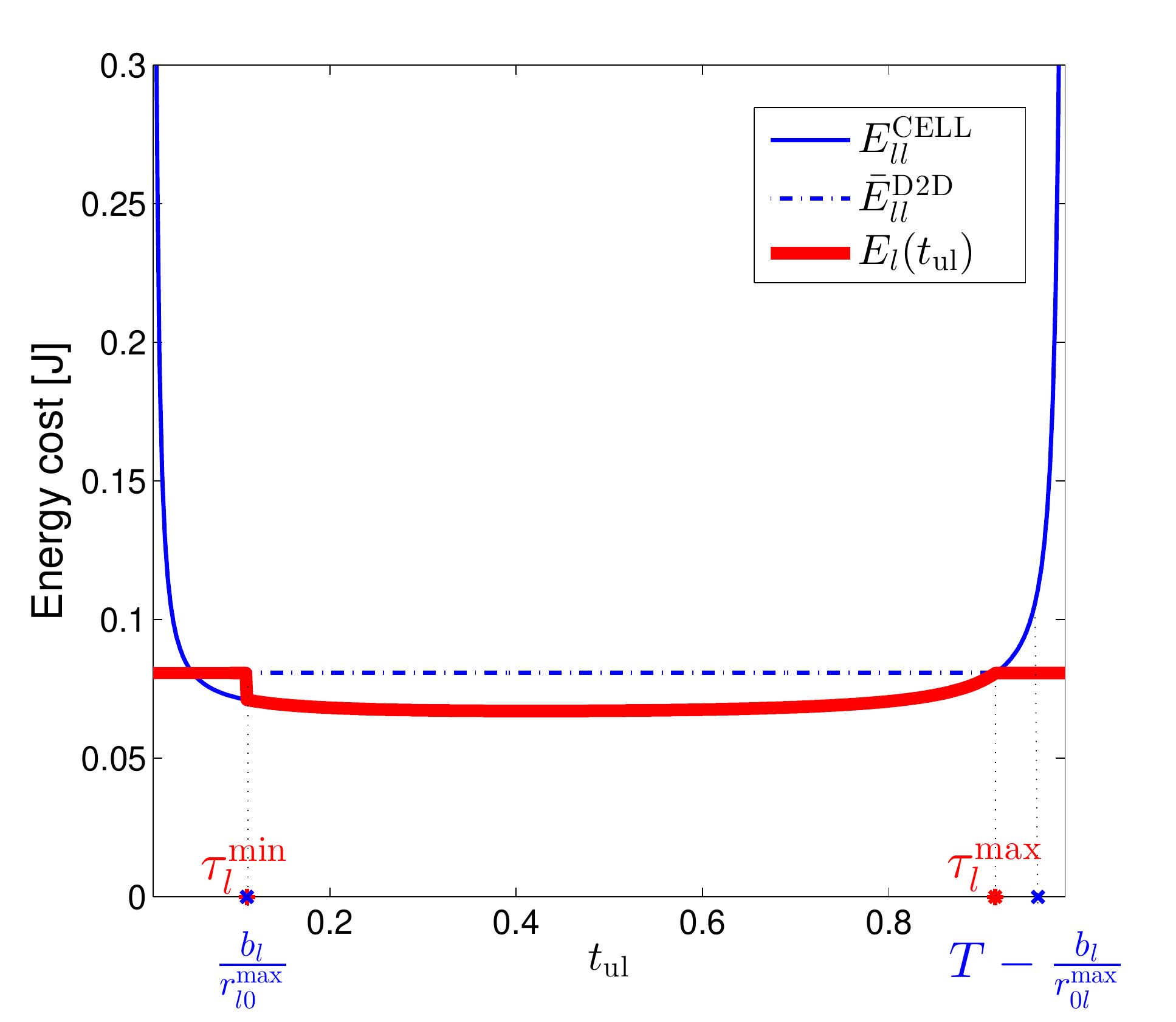}
        }
    \end{center}
\caption{Deriving $E_l(t_{\rm ul})$ under both the FO-UE case (a) and the FO-SE case (b). The interval $\Delta_l=[\tau_l^{\rm min},\tau_l^{\rm max}]$ depends on whether and at which points the graphs of  $E_{ll}^{\rm CELL}(t_{\rm{ul}})$ and $\bar{E}_{ll}^{D2D}$ intersect, but also on the interval $[\frac{b_l}{r_{l0}^{\max}}, T -\frac{b_l}{r_{0l}^{\max}} ]$ for power feasibility in cellular mode.
The two intervals coincide when the graphs intersect outside the power feasible interval (a).
 $\Delta_l \subset [\frac{b_l}{r_{l0}^{\max}}, T -\frac{b_l}{r_{0l}^{\max}} ]$ if at least one of the two possible intersection points is within the power feasible interval (b). $\Delta_l=\emptyset$ if the two graphs never intersect.}\label{fig: Example}
\end{figure}

Figures~\ref{User_ToyEx} and~\ref{Ntwrk_ToyEx} show the minimum energy-cost of each pair, obtained as illustrated in Fig.~\ref{fig: Example}, and the corresponding $F(t_{\rm ul})$ of a simple network with three user pairs, under FO-UE or FO-SE, respectively.
Note that function $F(t_{\rm ul})$ is non-convex on the interval  $[0,T]$. The following lemma establishes the key property of
  $F(t_{\rm ul})$, which will be useful later when solving the problem formulated in~\eqref{PB_ML_2}.

%%%%%%%%%%%%%%%%%%%%%%%%%%%%%%%%%%%%%%%%%%%%%%%%%%%%%%
\begin{lemma} \label{lemma: function F}
\begin{enumerate}[a)]
\item[(a)] In the FO-SE scenario, $F(t_{\rm ul})$ is a piecewise convex function.
\item[(b)] In the FO-UE scenario, $F(t_{\rm ul})$ is a piecewise decreasing function.
\end{enumerate}
\end{lemma}

\begin{proof}
For each pair-$l$, by its definition we have that if $\Delta_l = \emptyset$, $E_{ll}^{\rm CELL}(t_{\rm ul})$ is a constant value on $[0,T]$, otherwise,   $E_{ll}^{\rm CELL}(t_{\rm ul})$ is a constant value or $+\infty$ in the two intervals $[0,\tau_l^{\rm min})$ and $(\tau_l^{\rm max},T]$. During the interval $[\tau_l^{\rm min},\tau_l^{\rm max}]$,  in the FO-SE case $E_{ll}^{\rm CELL}(t_{\rm ul})$ is given by~\eqref{SE-MS}, which a convex function of $t_{\rm ul}$; while in the FO-UE case, $E_{ll}^{\rm CELL}(t_{\rm ul})$ is given by~\eqref{UE-MS}, which is a monotonically decreasing function in $t_{\rm ul}$. The function $F(t_{\rm ul})$ is obtained as the sum of $E_{ll}^{\rm
    CELL}(t_{\rm ul})$ of all $L$ pairs. Hence, the whole interval $[0,T]$ is divided into $J \leq 2L+1$ adjacent intervals.
In the FO-UE case, $F(t_{\rm ul})$ is the sum of constants and convex functions in each interval, which makes $F(t_{\rm ul})$  piecewise convex. In the FO-UE case, on the other hand, $F(t_{\rm ul})$ is the sum of constants and monotonically decreasing functions in each interval, which makes $F(t_{\rm ul})$ piecewise decreasing.
\end{proof}
%%%%%%%%%%%%%%%%%%%%%%%%%%%%%%%%%%%%%%%%%%%%%%%%%%%%%%

\begin{figure}[h]
     \begin{center}
        \subfigure[FO-UE three links: energy-cost functions.]{
            \label{User_ToyEx: sub1}
            \includegraphics[scale=0.33]{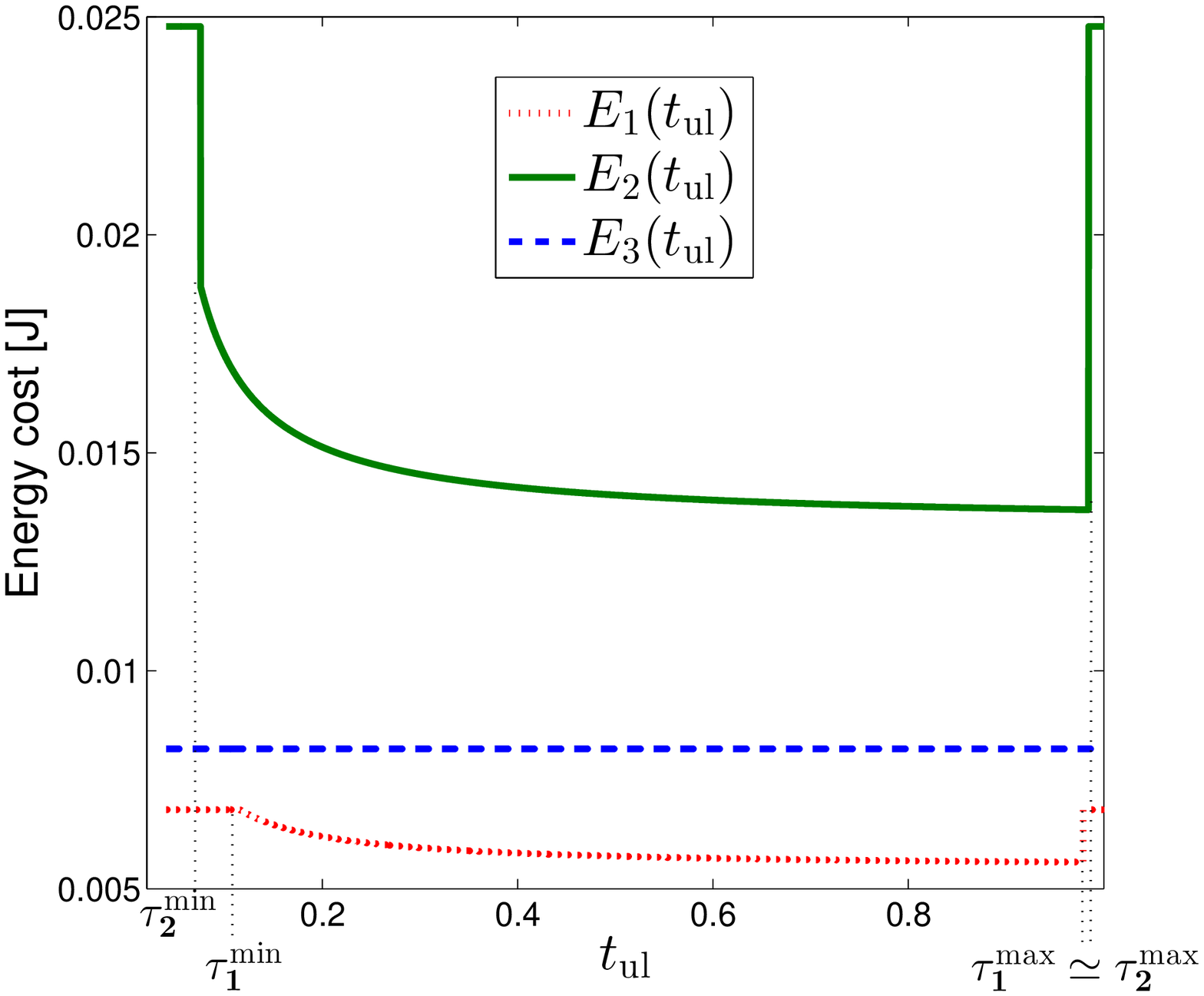}
        }
        \subfigure[FO-UE three links: $F(t_{\rm ul})$.]{
           \label{User_ToyEx: sub2}
           \includegraphics[scale=0.33]{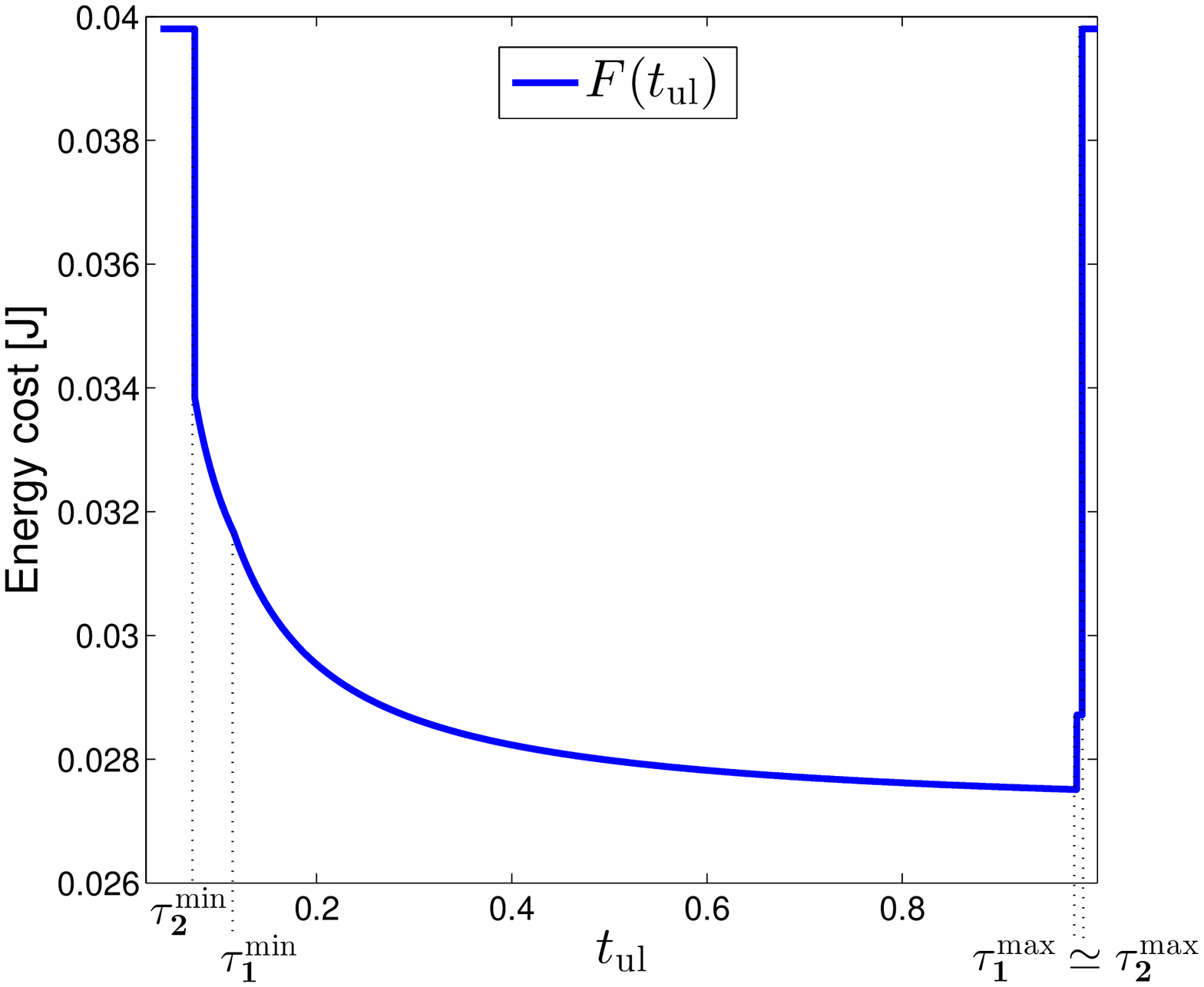}
        }
    \end{center}
\caption{FO-UE problem: minimum energy-cost functions of three transmitter-receiver pairs (a), and their corresponding sum $F(t_{\rm ul})$ (b).}\label{User_ToyEx}
\end{figure}

\begin{figure}[h]
     \begin{center}
        \subfigure[FO-SE three links: energy-cost functions.]{
            \label{Ntwrk_ToyEx: sub1}
            \includegraphics[scale=0.33]{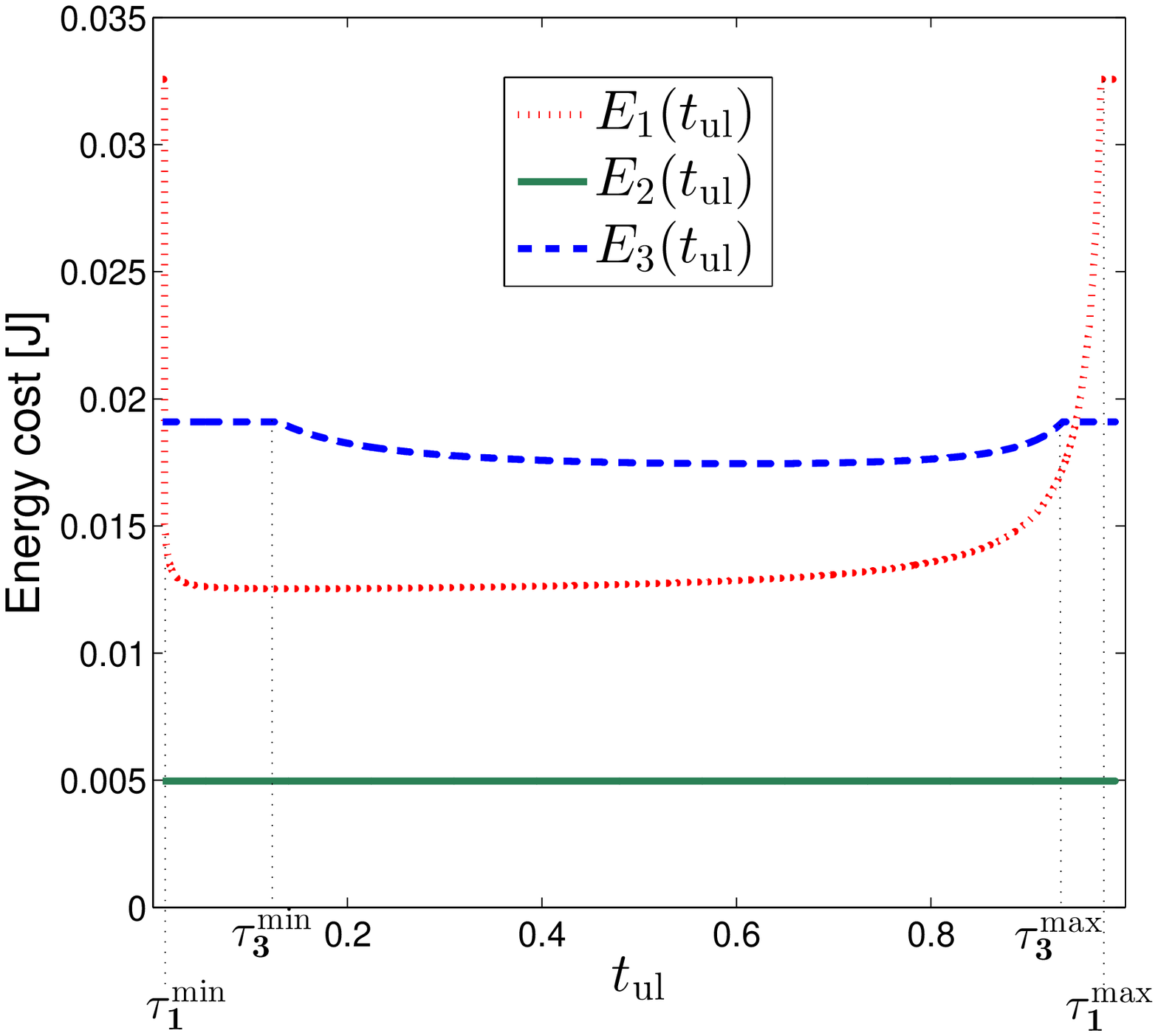}
        }
        \subfigure[FO-SE three links: $F(t_{\rm ul})$.]{
           \label{Ntwrk_ToyEx: sub2}
           \includegraphics[scale=0.33]{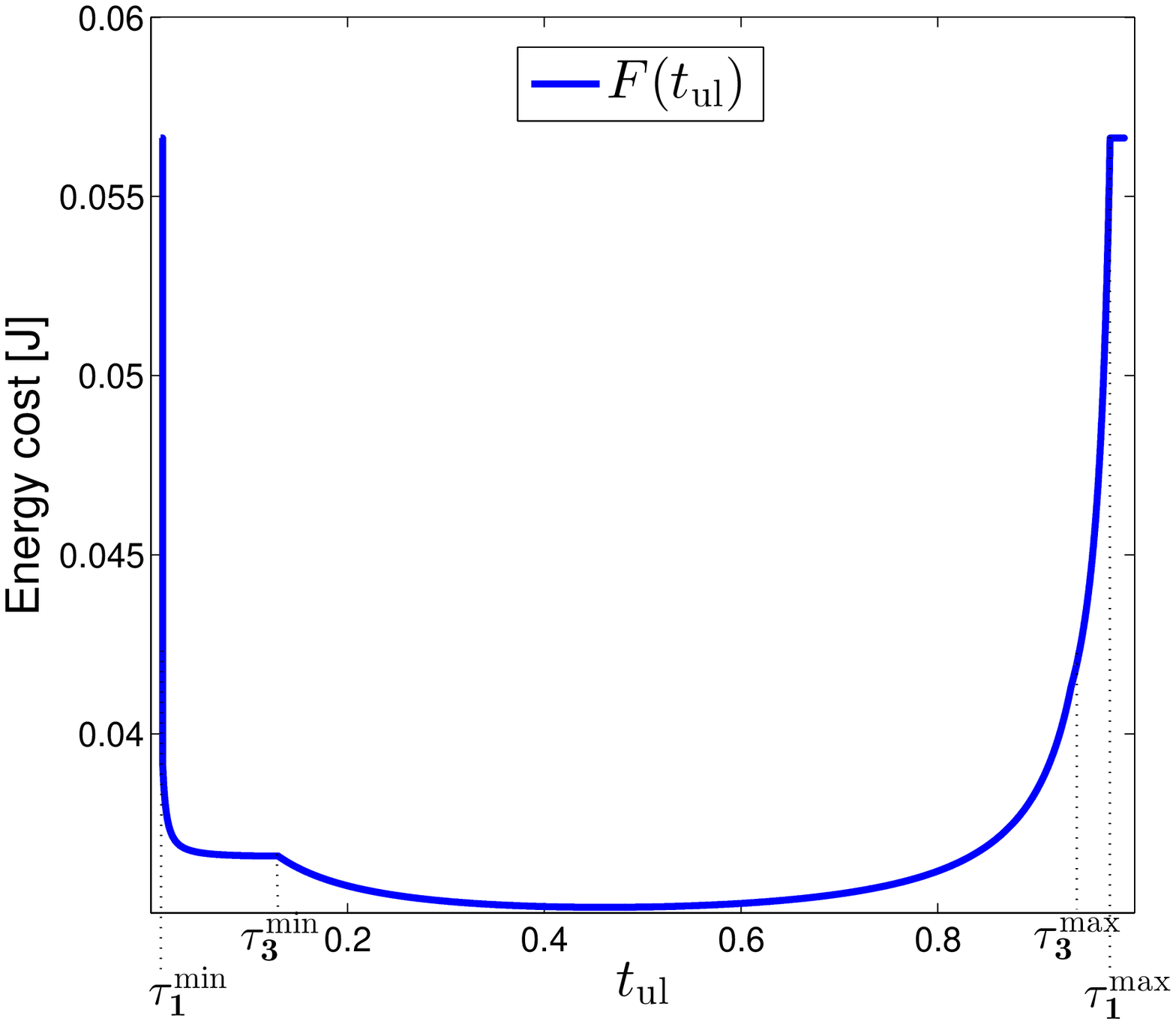}
        }
    \end{center}
\caption{FO-SE problem: minimum energy-cost functions of three transmitter-receiver pairs (a), and their corresponding sum $F(t_{\rm ul})$ (b).}\label{Ntwrk_ToyEx}
\end{figure}

Based on Lemma~\ref{lemma: function F}, the optimal solution to~\eqref{PB_ML_2} can be computed efficiently.

%%%%%%%%%%%%%%%%%%%%%%%%%%%%%%%%%%%%%%%%%%%%%%%%%%%%%%
\begin{prop}
\begin{enumerate}[a)]
\item In the FO-SE case, let $\cup_{j=1}^{J} \Gamma_j$ be a partition of $[0, T]$ induced by the points $\{ \tau_1^{\min}, \tau_1^{\max}, \dots, \tau_{L}^{\min}, \tau_{L}^{\max}\}$. Then, the optimal UL time allocation $t_{\rm ul}^\star$ can be found by solving \emph{at most} $2L - 1$ single-variable convex  optimization problems of the form
\[\underset{t_{\rm ul} \in \Gamma_j}{\text{minimize}}
\quad  \sum_{l\in {\mathcal L}} E_{l}(t_{\rm ul})
\]
\item In the FO-UE case,
$t_{\rm ul}^\star \in \{ \tau_1^{\rm max}, \tau_2^{\max}, \dots, \tau_{L}^{\max} \}$. Moreover, if $\underset{l}{\max}\{\tau_l^{\rm min}\} \leq \underset{l}{\min} \{\tau_l^{\rm max}\}$, then $t_{\rm ul}^\star = \underset{l}{\min} \{\tau_l^{\rm max}\} =\underset{l} {\min} \{T -\frac{b_l}{r_{0l}^{\max}}\}$.
\end{enumerate}
\end{prop}

\begin{proof}
Since for each user pair-$l$, $E_{ll}^{\rm CELL}(t_{\rm ul})$  reaches its maximum value  during the two intervals $[0,\tau_l^{\rm min})$ and $(\tau_l^{\rm  max},T]$, thus $F(t_{\rm ul})$ achieves its maximum value in the two intervals $\Gamma_1=\left[ 0, \min_l\{\tau_l^{\rm min}\}\right]$ and $\Gamma_J=\left[\max_l\{\tau_l^{\rm max}\}, T\right]$. Hence,  $t_{\rm ul}^\star$  is not in $\Gamma_1$ and $\Gamma_J$, but must be found in one of the remaining (at most) $2L-1$ intervals.
By Lemma~\ref{lemma: function F}, in the FO-SE case, $F(t_{\rm ul})$ is piecewise convex. Hence, its global minimum can be found among its $2L-1$ local minima in each interval.
In the FO-UE case, from Lemma~\ref{lemma: function F}, we know that $F(t_{\rm ul})$ is piecewise decreasing. Thus, its global minimum can be found in the set $\cup_{l\in \mathcal{L}} \{\tau_l^{\rm min}, \tau_l^{\rm max}\}$. However, for each $\tau_l^{\rm min}$, there is at least one component $E_l(t_{\rm ul})$ in the sum defining $F(t_{\rm ul})$, that decreases for $t_{\rm ul} \geq \tau_l^{\rm min}$. Therefore, the global minimum can only be found in the set $\cup_l \{ \tau_l^{\rm  max}\}$. Furthermore, if $\underset{l}{\max} \{\tau_l^{\rm min}\} \leq \underset{l}{\min} \{\tau_l^{\rm max}\}$, then $t_{\rm ul}^\star = \underset{l}{\min} \{\tau_l^{\rm max}\}$.
\end{proof}
%%%%%%%%%%%%%%%%%%%%%%%%%%%%%%%%%%%%%%%%%%%%%%%%%%%%%%

Given the optimal solution $t_{\rm ul}^\star$ to problem~\eqref{PB_ML_2}, the
optimal mode selection vector $\mathbf{m}^\star$ of
Problem~\eqref{PB_ML} is then given by setting, $\forall l \in \mathcal{L}$:
\[
m_l^\star = \begin{cases}
				0 & \text{if } E_{ll}^{\rm CELL}(t_{\rm ul}^\star) \leq  \bar{E}^{\rm D2D}_{ll}\\
				1 & \text{otherwise},
\end{cases}
\]
and the corresponding optimal transmission powers are derived as in~\eqref{eq: opt powers}.

%%%%%%%%%%%%%%%%%%%%%%%%%%%%%
%%%%% INTERFERENCE CASE %%%%%
\section{Minimum-Energy Mode Selection with D2D Resource Sharing} \label{Sec: Interf}
To increase the cell capacity and the spectral efficiency of the system, we now consider the D2D resource sharing (RS) strategy, where all communications in D2D mode are assigned the same channel resource. We show that the optimal mode selection and the optimal power/time allocation are harder to compute when interference is considered among multiple communications. Nevertheless, we develop a combinatorial optimization algorithm that is guaranteed to find the optimal solution, and often does so very quickly. This optimal method is complemented by a fast heuristic, suitable for real-time implementation under practical signalling constraints.

\subsection{Optimal resource allocation via mixed-integer nonlinear programming}
The effect of interference on the energy consumption in D2D mode appears explicitly in \eqref{eq: energy}.
Due to the fixed time allocation $T$ (see Lemma~\ref{lemma: Energy cost D2D}), minimizing the energy consumption of D2D communications is equivalent to minimizing the transmission powers. To meet the minimum rate requirement in~\eqref{eq: rateConstraint2}, the transmission power of any pair-$l$ in D2D mode must be such that
\[
 p_{ll} \geq \left[\exp\left( \frac{b_l}{WT}\right) -1 \right] \frac{\sigma^2 + I_l}{G_{ll}}
 %+ \left[\exp\left( \frac{b_l}{WT}\right) -1 \right] \frac{I_l}{G_{ll}} .
\]
Introducing $\gamma_l^{\rm{tgt}} = \left[ \exp \left(\frac{b_l}{WT} \right) -1 \right]$ as the target signal-to-interference-plus-noise ratio (SINR) required to satisfy the session rate requirement of pair-$l$, $\eta_l = \frac{\gamma_l^{\rm{tgt}}\sigma^2}{G_{ll}}$, and $h_{lj} = \gamma_l^{\rm{tgt}} \frac{G_{jl}}{G_{ll}}$, we can re-write this inequality as
\begin{equation} \label{eq: P_I}
p_{ll} \geq  \eta_l  + \sum_{j\neq l} p_{jj} h_{lj}.
\end{equation}

The joint mode selection and power/time allocation problem can now be formulated as the following MINLP problem:
\begin{subequations} \label{PB_ML_I}
\begin{align}
   \underset{\mathbf{m} , t_{\rm ul}, p_{ll} }{\text{minimize}}
        & \quad  \sum_{l=1}^L (T p_{ll}) m_l + E_{ll}^{\rm CELL}(t_{\rm ul}) (1-m_l) \\
    \text{subject to}
        &\quad  \frac{b_l}{r_{l0}^{\max}} - T m_l \leq t_{\rm ul} \leq T -\frac{b_l}{r_{0l}^{\max}} + T m_l, &&\forall l, \label{C1: PB_ML_I} \\
        &\quad  \left[ \eta_l  +  \sum_{j\neq l} p_{jj} h_{lj} \right] - C(1- m_l) \leq p_{ll}, &&\forall l, \label{C2: PB_ML_I}\\
                &\quad  0\leq p_{ll} \leq p_l^{\rm max} m_l, &&\forall l\\
        &\quad \mathbf{m} \in \{0,1\}^L, \quad t_{\rm ul} \in [0,T].
\end{align}
\end{subequations}
Here, $\mathbf{m}$ is the mode selection vector defined in~\eqref{m}. Constraint~\eqref{C1: PB_ML_I} ensures power feasibility for pairs in cellular mode, while~\eqref{C2: PB_ML_I} guarantees that the rate requirement is satisfied for each pair in D2D mode. The constant $C$ in~\eqref{C2: PB_ML_I} is a large number ($C=\max_l \{\eta_l + \sum_{l} p_{ll}^{\max}\}$, for example) ensuring that the constraint is only enforced for users in D2D mode.
The expression for $ E_{ll}^{\rm CELL}(t_{\rm ul})$ in the objective function is either~\eqref{SE-MS} or~\eqref{UE-MS}, depending on whether we are interested in the RS-SE or RS-UE problem, respectively.

Problem~\eqref{PB_ML_I} belongs to the class of mixed boolean-convex problem, where for each fixed $\mathbf{m} \in \{0,1\}^L$ the objective function is convex in the continuous variables.
In general, MINLPs are NP-hard problems \cite{LeLe2012}, combining the combinatorial difficulty of optimizing over discrete variable sets, with the challenges of handling nonlinear functions. Its solution time grows exponentially with the problem dimension. If $L$ is very small (\emph{i.e.}, smaller than 15), it can be solved exactly by exhaustive enumeration of the $2^L$ possible mode selection vectors.
However, realistic cellular networks might consist of a large number of user pairs. For this reason, we propose an algorithm, based on a B\&B strategy (\emph{e.g.}, \cite{dakin1965tree, Gupta1985}), that often finds the optimal solution to~\eqref{PB_ML_I} in a much more efficient way than the exhaustive search.

Due to space restrictions, in the sequel we will focus on algorithms that solve the RS-UE problem, considering that the mobile devices are the most energy-sensitive component of the network. However, the same approaches can be applied to the RS-SE case.

%%%%%%%%%%%%%%%%%%
\subsection{A branch-and-bound approach for finding the optimal solution} \label{opt_sol}
We propose to solve the MINLP problem~\eqref{PB_ML_I} using a B\&B approach, where possible mode selection vectors $\mathbf{m}$ are explored through a binary tree. Each node of the tree (except the root) represents a subproblem where one of the mode selection variables $m_l$, with $ l\in \mathcal{L}$, is set to either $0$ or $1$. Each branch corresponds to a subset of the possible mode selection vectors.

%%%%%%%%%
Before proceeding, it is convenient to introduce a definition and two useful propositions on the feasibility of the mode selection vector $\mathbf{m}$.  Let $\mathbf{H}$  be the non-negative matrix  with entries $H_{lj}= h_{lj}$ if $l \neq j$, and zero otherwise, and the vectors $\mathbf{p} = (p_{ll}, \forall l \in \mathcal{L})^T$, $\mathbf{p}^{\max} = (p_{l}^{\max}, \forall l \in \mathcal{L})^T$ and $ \boldsymbol{\mathbf{\eta}}= (\eta_l, \forall l \in \mathcal{L})^T$. 
For each mode selection vector $\mathbf{m}$, we can define the corresponding set of pairs assigned to D2D mode and to cellular mode as $\mathcal{D}_\mathbf{m}$ and  $\mathcal{C}_\mathbf{m}$, respectively. 
Let $\mathbf{A_m}$ denote the $L \times |\mathcal{D}_\mathbf{m}|$ incidence matrix, which is formed by removing the $l$-th column from the $L \times L$ identity matrix if $m_l=0$.
We define $\mathbf{H_m}=\mathbf{A}_{\mathbf{m}}^T \mathbf{H}\mathbf{A_m}$, $\mathbf{p_m}=\mathbf{A}_{\mathbf{m}}^T\mathbf{p}$, and $\boldsymbol{\mathbf{\eta_m}}=\mathbf{A}_{\mathbf{m}}^T \boldsymbol{\mathbf{\eta}}$. The constraint~\eqref{C1: PB_ML_I} can be written in matrix form as
\begin{equation} \label{Ineq_1}
(\mathbf{I_m} - \mathbf{H_m})  \mathbf{p_m} \geq \boldsymbol{\mathbf{\eta_m}} \quad  \text{and} \quad \mathbf{p_m} \leq \mathbf{p}_{\mathbf{m}}^{\rm max},
\end{equation}
where  the inequalities are component-wise, $\mathbf{I_m}=\mathbf{A}_{\mathbf{m}}^T\mathbf{A}_{\mathbf{m}}$, and $\mathbf{p}_{\mathbf{m}}^{\rm max} = \mathbf{A}_{\mathbf{m}}^T \mathbf{p}^{\rm max}$.

Matrix $\mathbf{H_m}$ has strictly positive off-diagonal elements, and we can assume that it is irreducible because we do not consider totally isolated groups of pairs that do not interact with each other. Let $\rho(\mathbf{H_m})$ denote the largest real eigenvalue of matrix $\mathbf{H_m}$.
From the Perron-Frobenius theorem~\cite{1985:matrix}, we have the following well-known proposition:

%%%%%%%%%%%%%%%% PROPOSITION %%%%%%%%%%%%%%%%%%
\begin{prop} [\cite{PC2008} Chapter~2] \label{Prop: PF}
For a given mode selection vector $\mathbf{m}$, the necessary and sufficient condition for the existence of a positive $ \mathbf{p_m} $ to solve inequality $(\mathbf{I_m} - \mathbf{H_m})  \mathbf{p_m} \geq  \boldsymbol{\mathbf{\eta_m}}$ is that
\begin{equation}\label{PF}
\rho(\mathbf{H_m})<1.
\end{equation}
Moreover, $\mathbf{p^\star_m} = (\mathbf{I_m} - \mathbf{H_m})^{-1}\boldsymbol{\mathbf{\eta_m}}$ is its component-wise minimum solution.
\end{prop}
%%%%%%%%%%%%%%%%%%%%%%%%%%%%%%%%%%%%

Proposition~\ref{Prop: PF} provides an easy condition to verify if a mode selection vector $\mathbf{m}$ is feasible.
%%%%%%%%%%%%%%%% DEFINITION %%%%%%%%%%%%%%%%%%
\begin{defn}[Feasible mode selection vector]
A \textit{mode selection vector} $\mathbf{m}$ is  feasible if both condition~\eqref{PF} and $(\mathbf{I_m} - \mathbf{H_m})^{-1}\boldsymbol{\mathbf{\eta_m}} \leq \mathbf{p}_{\mathbf{m}}^{\rm max}$ are verified.
\end{defn}
%%%%%%%%%%%%%%%% PROPOSITION %%%%%%%%%%%%%%%%%%
\begin{prop} [\cite{Feasibility_Ephremides2006}] \label{Prop: InfeasibleSubSet}
If $\mathbf{m}$ is not feasible, every other mode selection  vector $\mathbf{\tilde m}$, such that the set $\{ l \in \mathcal{L}: m_l=1 \} \subseteq  \{l \in \mathcal{L}: \tilde{m}_l=1 \}$, is not feasible.
\end{prop}
%%%%%%%%%

The main idea of B\&B is to only explore branches of the binary tree that have the potential to produce better solutions than the best solution found so far. This is done by computing upper and lower bounds on the optimal value at each node. If the lower bound of a node is larger than the current upper bound, then there is no need to explore its branches.

%Conceptually, we start with an initial feasible solution and use its objective value as an upper bound on the optimal cost. We then explore the search tree using a branching rule that determines which mode selection variable should be fixed next. A tree exploration strategy determines if we should fix the considered variable to zero or to one first. Finally, upper and lower bounds are computed at each node in the search tree and branches that cannot be optimal are cut. 

To achieve a good performance of B\&B, it is essential to select the branching rule and tree exploration strategies carefully, and to have efficient methods for computing good upper and lower bounds~\cite{floudas1995nonlinear}. 
In our implementation, we have made the following choices:
\begin{enumerate}
	\item \emph{Initial upper bound.} By assumption, letting all pairs communicate in cellular mode, \emph{i.e.} setting $m_l=0$ for all $l$, is always feasible. We therefore take the corresponding energy cost as an initial upper bound on the optimal cost, computed by solving 
\begin{equation} \label{PB: All cell}
\begin{split}
    \underset{ t_{\rm ul}\in[0,T]} {\text{minimize}}
        & \quad  \sum_{l \in \mathcal{L}} E_{ll}^{\rm CELL}(t_{\rm ul})\\
    \text{subject to}
        &\quad  \max_{l  \in \mathcal{L}}\{\frac{b_l}{r_{l0}^{\max}}\}  \leq t_{\rm ul} \leq  \min_{l  \in \mathcal{L}}\{ T -\frac{b_l}{r_{0l}^{\max}} \},
\end{split}
\end{equation} 	
where $E_{ll}^{\rm CELL}(t_{\rm ul})$ is given by \eqref{UE-MS} in the RS-UE case.
	\item \emph{Branching rule.} The branching rule selects the next variable to fix at each node of the tree. The goal is to identify the branching variable that changes the problem the most, either to quickly detect branches that can be cut, or to significantly improve the current solution.
	
Our branching strategy for the RS-UE problem is based on first solving the FO-UE problem described in \S~\ref{Sec: NoInterf}. The corresponding optimal mode selection vector $\mathbf{m}^{\rm FO}$ reveals the set ${\mathcal D}^{\rm FO} = \{ l \;\vert\; m^{\rm FO}_l=1\}$ of pairs that prefer to communicate in D2D mode in an interference-free environment.  For each link-$l\in {\mathcal D}^{\rm FO}$ we define a measure of its strength $s_l= \sum_{i\in {\mathcal D}^{\rm FO}, \, i\neq l} G_{li}/G_{ll}$. This measure attempts to account for both the interference that pair-$l$ produces on the shared frequency resource and its own direct gain. The branching rule first selects variables in ${\mathcal D}^{\rm FO}$ in order of decreasing $s_l$, and then considers the remaining variables in an arbitrary order. This rule exploits the fact that the optimal solution of the RS-UE problem is often close to the FO-UE optimal, with a few differences in the pairs in D2D mode, due to interference; pairs that prefer to be in cellular mode in FO will also tend to prefer cellular mode in RS. By fixing pairs with high interference strength first, we increase the likelihood of finding infeasible solutions quickly. Once an infeasible mode selection is found, we can make use of Proposition~\ref{Prop: InfeasibleSubSet} and discard all branches below the current node in the search tree.
	\item \emph{Tree exploration strategy.} Once a branching variable $m_l$ has been selected by the branching rule, the tree exploration strategy determines if the child node to investigate next should be in cellular or D2D mode (that is, have $m_l=0$ or $1$). We always prefer to try to allocate users to D2D mode first (that is, set the branching variable $m_l=1$ first).
	\item \emph{Upper and lower bounds.} 
	Each node of the search tree corresponds to a partial mode selection vector with some components set to $1$ or $0$, while others are still undetermined. We define the two sets of pairs corresponding to the determined (fixed) and undetermined variables, as  $\mathcal{F}$ and $\mathcal{U}$,  respectively.
	
	When we consider a node, we first verify if the fixed variables form a feasible mode selection vector in the sense of Proposition~\ref{Prop: PF}.	
	If they do not, then no bounds are computed and the node and all branches below it are disregarded. Otherwise, upper and lower bounds are obtained as the sum of the minimum energy cost of the pairs in $\mathcal{F}$ and, respectively, an upper and a lower bound of the energy cost of the pairs in $\mathcal{U}$.

To compute an upper bound, we assign all pairs in  $\mathcal{U}$ to cellular mode and we solve the problem formulation in~\eqref{PB: All cell} where $\mathcal{L}$ is replaced by $\mathcal{U}$.
To determine a lower bound, we compute the minimal energy cost of the unassigned pairs when they operate in full orthogonality (\emph{i.e.}, we solve the FO-UE problem over only the unassigned pairs). To understand why this is a lower bound, note that the computation is a proper relaxation of the MINLP where the UL times of fixed and unassigned pairs are allowed to differ, and when the unassigned links that end up in D2D mode do not suffer interference. Furthermore, we strengthen the lower bound by increasing the noise power of pairs in ${\mathcal U}$ by the interference that the transmitters fixed to D2D communication in $\mathcal{F}$ incur on them.
\end{enumerate}

Once achieved the optimal feasible mode selection vector and transmission time, the corresponding optimal powers are obtained as in~\eqref{eq: opt powers} for the cellular users, and as in Proposition~\ref{Prop: PF} for transmitters in D2D mode.

%%%%%%%%%%%%%%%%%%%%%%%%%%%%%%%%
The B\&B algorithm is guaranteed to find the optimal solution, and does so much faster than the exhaustive search, as shown in \S~\ref{Sec: results_B&B}. However, for large networks it can still have impractical running times. In addition, the optimization formulation assumes that all cross-gains between users are known, something that would require significant communication overhead. We therefore turn our attention to heuristics that can be run in real-time and do not assume  centralized knowledge of all the channel gains.

\subsection{Heuristic approach to achieve  a practical sub-optimal solution} \label{heuristic}
In this section, we present a heuristic algorithm that achieves a near-optimal solution to~\eqref{PB_ML_I} in a more practical and scalable way than the B\&B approach. Again, we focus on the UE case.

The key idea of this algorithm is to first determine an initial mode selection vector, together with the corresponding power/time allocation, and then improve this solution by means of a distributed power control algorithm based only on local measurements.
The heuristic algorithm is described in Algorithm \ref{algo1}, with the two main steps in the following:

%%%%%%  ALGORITHM  %%%%%%
\begin{algorithm}
\caption{Heuristic approach for RS-UE minimization}
\label{algo1}
\KwIn{ $(\gamma_l^{\rm{tgt}}, G_{ll}, G_{l0}, G_{0l}, \theta) \,\forall l \in \mathcal{L}$}
\KwOut{ $\mathbf{m}^\star, \mathbf{p}^\star$}
$( \mathbf{m}^{\rm FO}, t_{\rm ul}(\mathbf{m}^{\rm FO}), \mathbf{p}(\mathbf{m}^{\rm FO}) )\leftarrow $ solution to FO-UE problem\;
each $l \in  \mathcal{D}_{\mathbf{m}^{\rm FO}}$  acquires $E_{ll}^{\rm CELL}\left(t_{\rm ul}(\mathbf{m}^{\rm FO})\right)$ from the BS\;
$\mathbf{p}^{(0)}  \leftarrow \mathbf{p}(\mathbf{m}^{\rm FO}) $, \,
$\mathbf{m}^{(0)}  \leftarrow \mathbf{m}^{\rm FO} $, \, $k=0$\;
each $l \in  \mathcal{D}_{\mathbf{m}^{\rm FO}}$ computes $\gamma_l^{(0)}$\;

convergence $\leftarrow $ True\;
\While{convergence}
		{$\mathbf{m}^{(k+1)} \leftarrow \mathbf{m}^{(k)}$\;	
		\For {each $l \in \mathcal{D}_{\mathbf{m}^{(k)}}$}
		   	{$p_{ll}^{(k+1)}  = \frac{\gamma_l^{\rm{tgt}}}{\gamma_l^{(k)}} \, p_{ll}^{(k)}$\;		
			\If {$p_{ll}^{(k+1)} > \min \left\lbrace \frac{\theta}{T} E_{ll}^{\rm CELL}\left(t_{\rm ul}(\mathbf{m}^{\rm FO})\right), \, p_l^{\rm max} \right\rbrace$}				
			{$ m_l^{(k+1)} \leftarrow 1, $ \, $ \mathcal{D}_{\mathbf{m}^{(k+1)}} \leftarrow \mathcal{D}_{\mathbf{m}^{(k)}} \setminus \{l\} $\;}
			}	
		each $l \in  \mathcal{D}_{\mathbf{m}^{(k+1)}}$ computes $\gamma_l^{(k+1)}$\;
		\If{$ \gamma_l^{(k+1)} \geq \gamma_l^{\rm{tgt}}, \forall l \in  \mathcal{D}_{\mathbf{m}^{(k+1)}}  $}	
		 {convergence $\leftarrow$ False\;}	
		 $\mathbf{p}^\star \leftarrow \mathbf{p}^{(k+1)},$ \,
		$\mathbf{m}^\star \leftarrow \mathbf{m}^{(k+1)}$\;}
		\end{algorithm}
%%%%%%%%%%%%%%%%%%%%%%%%%%%%%%%%%%%%%%%%%%%%%%%%%%%%%%%%%%%

\begin{enumerate}
\item \textit{Initial phase}:
We adopt the optimal solution to the FO-UE problem in \S~\ref{Sec: NoInterf} as the initial solution, denoted by $\left( \mathbf{m}^{\rm FO}, t_{\rm ul}(\mathbf{m}^{\rm FO}), \mathbf{p}(\mathbf{m}^{\rm FO}) \right)$.
The FO-UE problem is solved by the BS.
For each pair-$l \in \mathcal{D}_{\mathbf{m}^{\rm FO}}$ (\emph{i.e.}, assigned to D2D mode), the BS also computes the energy it would consume if in cellular mode, that is $E_{ll}^{\rm CELL}\left(t_{\rm ul}(\mathbf{m}^{\rm FO})\right)$ from \eqref{UE-MS}, and broadcasts $\mathbf{m}^{\rm FO}$ and $E_{ll}^{\rm CELL}\left(t_{\rm ul}(\mathbf{m}^{\rm FO})\right)$ to each Tx-$l \in \mathcal{D}_{\mathbf{m}^{\rm FO}}$.

The initial mode selection vector $\mathbf{m}^{\rm FO}$ is obtained under the assumption of no interference among the D2D pairs. However, under the RS scenario, all the D2D pairs share the same channel, thus $\mathbf{m}^{\rm FO}$ can be energy inefficient, or even infeasible, due to the interference. Therefore, a distributed power control algorithm is then executed by the D2D pairs to find a feasible and more energy-efficient solution:

\item \textit{Iterative distributed power control for D2D pairs}:
Using the iterative power control method originally proposed by Foschini and Miljanic in~\cite{Foschini93}, each Tx-$l$ in D2D mode can achieve its target SINR $\gamma_l^{\rm tgt}$ by updating its transmit power as follows
\begin{equation} \label{DPC}
p_{ll}^{(k+1)} = \frac{\gamma_l^{\rm tgt}}{\gamma_l^{(k)}} p_{ll}^{(k)},
\end{equation}
where $\gamma_l^{(k)}$ is the perceived SINR for pair-$l  \in \mathcal{D}_{\mathbf{m}^{\rm FO}}$ in iteration-$k$, defined as $\gamma_l^{(k)}= \frac{p_{ll}^{(k)}G_{ll}}{\sigma^2 + \sum_{j\in\mathcal{D}_{\mathbf{m}^{\rm FO}},j\neq l} \, p_{jj}^{(k)}G_{jl}}$, and $\mathbf{p}^{(0)} = \mathbf{p}({\mathbf{m}^{\rm FO}})$.

To achieve a feasible mode selection vector and to further reduce the energy cost, some links in D2D mode need to switch to cellular mode. Specifically, pair-$l$ in D2D mode will switch to cellular mode if its transmit power level exceeds its maximum limit or if it is more energy efficient for it to communicate in cellular mode, that is,
\begin{equation} \label{P_max_D2D_2}
p_{ll}^{(k)} > \min \left\lbrace \frac{\theta}{T} E_{ll}^{\rm CELL}\left(t_{\rm ul}(\mathbf{m}^{\rm FO})\right), \, p_l^{\rm max} \right\rbrace,
\end{equation}
where we introduce the design parameter $\theta \geq 1$.
%%%%%%%%%%%%%%%%%REMARK THETA%%%%%%%%%%%%%
\begin{remark}  \label{remarkAlpha}
The selection of the parameter $\theta$ accounts for the following key aspects:
\begin{enumerate}
\item Mode switches incur additional signaling overhead between mobile devices and the BS to coordinate the re-allocation of radio resources. 
Hence, if $\theta > 1$ mode switches will occur only if they result in a significant energy gain.
\item  When a pair in D2D mode switches to cellular mode, it requires another orthogonal frequency channel for its transmission. The parameter $\theta$ can control this trade-off between channel reuse and energy consumption. Specifically, a large value of $\theta$ enforces more pairs to communicate in D2D mode and reuse the same channel, even if this comes at the cost of a higher energy consumption.
\item Once the final set of pairs assigned to cellular mode has been obtained, the new optimal UL transmission time cannot be larger than the initial $t_{\rm ul}(\mathbf{m}^{\rm FO})$, which means that the energy consumption of each pair in cellular mode could increase. Hence, by setting $\theta > 1$, we reserve a margin for the energy increase in cellular mode due to mode switches during the distributed power control process.
\end{enumerate}
\end{remark}
%%%%%%%%%%%%%%%%%%%%

During the power update~\eqref{DPC}, if Tx-$l$ finds that condition~\eqref{P_max_D2D_2} is fulfilled, it asks the the BS to switch it to cellular mode and to assign it an orthogonal frequency channel. Otherwise, it keeps updating its power according to~\eqref{DPC}. The BS keeps track of the pairs changing communication mode, and updates the mode selection vector. Specifically, if $\mathbf{m}^{(k)}$ denotes the mode selection vector at the $k$-th iteration of the algorithm, then
\begin{equation*}
m_l^{(k+1)}=
\begin{cases}
0, & \text{if } m_l^{(k)}=1, \text{and condition \eqref{P_max_D2D_2} is verified,} \\
m_l^{(k)}, & \text{otherwise},
\end{cases}
\end{equation*}
where the initial mode selection vector $\mathbf{m}^{(0)}=\mathbf{m}^{\rm FO}$.

This power control algorithm converges to the minimum power levels that the user pairs remaining in D2D mode need, to fulfill the rate requirement. Once the algorithm converges, the BS recomputes the optimal power/time allocation for the user pairs in cellular mode,  broadcasting this information before the data transmissions take place.
\end{enumerate}

%%%%%%%%%%%%%%%%%%%%
%%%%%% RESULTS %%%%%
\section{Simulations and discussion} \label{Sec: results}
This section presents simulation results that validate our theoretical findings and evaluate our proposed algorithms. All simulations consider a single cell with a BS, equipped with an omnidirectional antenna, positioned in the center of the cell. The simulation parameters listed in Table~\ref{Table1} are chosen to represent an urban LTE deployment and are used throughout.

%%%%%%%%%%%%%%%%%%%%%%%%%
\begin{table}[ht]
\centering
\begin{center}
\caption{Simulation parameters of the system under study}
\label{Table1}
\end{center}
\begin{tabular}[h]{|l|l|}
\hline
\emph{\textbf{Parameter}}     &       \emph{\textbf{Value}} \\ %[4ex]
\hline
Carrier Frequency &	1 GHz\\
\hline
Cell Radius  &	500 m\\
\hline
Frequency channel bandwidth  ($W$)  &   5 MHz     \\
\hline
Noise Power ($\sigma^2$)         & -174 dBm/Hz \\
\hline
Path-loss coefficient    ($\alpha$)      & 4 \\
\hline
Path gain at reference distance of 1m ($G_0$)      &  $5.7\cdot 10^{-4}$\\
\hline
Max Tx Power for the BS ($ p_{0}^{\rm max}$) & 40 W \\
\hline
Max Tx Power for transmitter-$l$ ($ p_{l}^{\rm max}$) & 0.25 W\\
\hline
Time Frame duration ($T$) & 1 time unit\\
\hline
\end{tabular}
\end{table}

%%%%%%%%%%%%%%%%%%%%%%%%%
\subsection{Single link analysis: geometrical interpretation of the optimal mode selection policy} \label{Res: sinlge link}
We begin by developing a geometrical interpretation of the optimal mode selection policy for the single link case, under the assumption that the channel gains follow a conventional path loss model $G_{ij}=G_0 D_{ij}^{-\alpha}$,  where $D_{ij}$ is the physical distance between Tx-$i$ and Rx-$j$, $G_{0}$ is the path gain at a reference distance of 1m, and $\alpha $ is the path-loss exponent.
Recall that D2D communication is preferable when $\bar{E}^{\rm D2D}_{ll}(T) \leq E^{\rm CELL}_{ll}(t_{\rm{ul}}^{\star})$.

Let us first study the user-energy objective, in which case $E^{\rm D2D}_{ll}(T)$ and $E^{\rm CELL}_{ll}(t_{\rm ul}^{\star})$ are given by (\ref{E_user}) and (\ref{E_D2D}), respectively. Using the path-loss model, we can transform the mode selection policy into the following equivalent condition in terms of the distances between the transmitter, receiver and BS:
\begin{align}
	D_{ll} &\leq \underbrace{\left( \frac{(e^{b_l/WT}-1)T}
		{(e^{b_l/Wt_{\rm ul}^{\star}}-1) t_{\rm ul}^{\star}}\right)^{-1/\alpha}}_{\kappa(D_{0l})}D_{l0} = \kappa(D_{0l}) D_{l0}. \label{eqn:geometrical_interpretation}
\end{align}
Note that $\kappa$ depends on $D_{0l}$, since $D_{0l}$ affects $r_{0l}^{\max}$ and thereby $t_{\rm ul}^{\star}=T-b_l/r_{0l}^{\max}$. Thus, even though we are neglecting the energy cost for the DL transmission, $D_{0l}$ still influences the optimal mode selection.
To characterize the region where D2D mode is preferable, we fix the position of Tx-$l$ (and therefore $D_{l0}$). We then vary the position of Rx-$l$ along a circle centred at the BS, thus keeping $D_{0l}$ and $\kappa(D_{0l})$ constant.
Inequality~\eqref{eqn:geometrical_interpretation} now states that D2D mode is preferable when the distance between the transmitter and receiver of pair-$l$ is less than $\kappa(D_{0l})D_{l0}$. 
In other words, D2D mode is more energy efficient when Tx-$l$ is located in the arc defined by the intersection of the circle of radius $D_{0l}$ centered at the BS, and the disc of radius $\kappa(D_{0l})D_{l0}$ centered at Tx-$l$.
The D2D optimal area can be constructed by tracing out these arcs for various distances between Rx-$l$ and the BS, see Fig.~\ref{Res: SingleLink}. One can show that $\kappa(D_{0l})$, and thus the D2D optimal area, decreases as Rx-$l$ gets closer to the BS.
	
	Fig.~\ref{Res: SingleLink UE} illustrates the D2D-optimal area in red and the D2D power-feasible are in light blue. To ensure that Assumption~\ref{Ass1} is satisfied, we set $b_l$ equal to the maximum traffic rate that can be supported when Rx-$l$ is on the cell boundary.
	
	 Although (\ref{eqn:geometrical_interpretation}) does not formally describe a disc around Tx-$l$, the D2D-optimal area is close to circular. The reason for this is the power imbalance between the user equipment and the BS, which makes $b_l/r_{0l}^{\max}$ very small,  $t_{\rm ul}^{\star}\approx T$ and $\kappa\approx 1$ practically independently of $D_{0l}$. 

%%%%%%
\begin{figure}[h]
\centering
\includegraphics[scale=0.25]{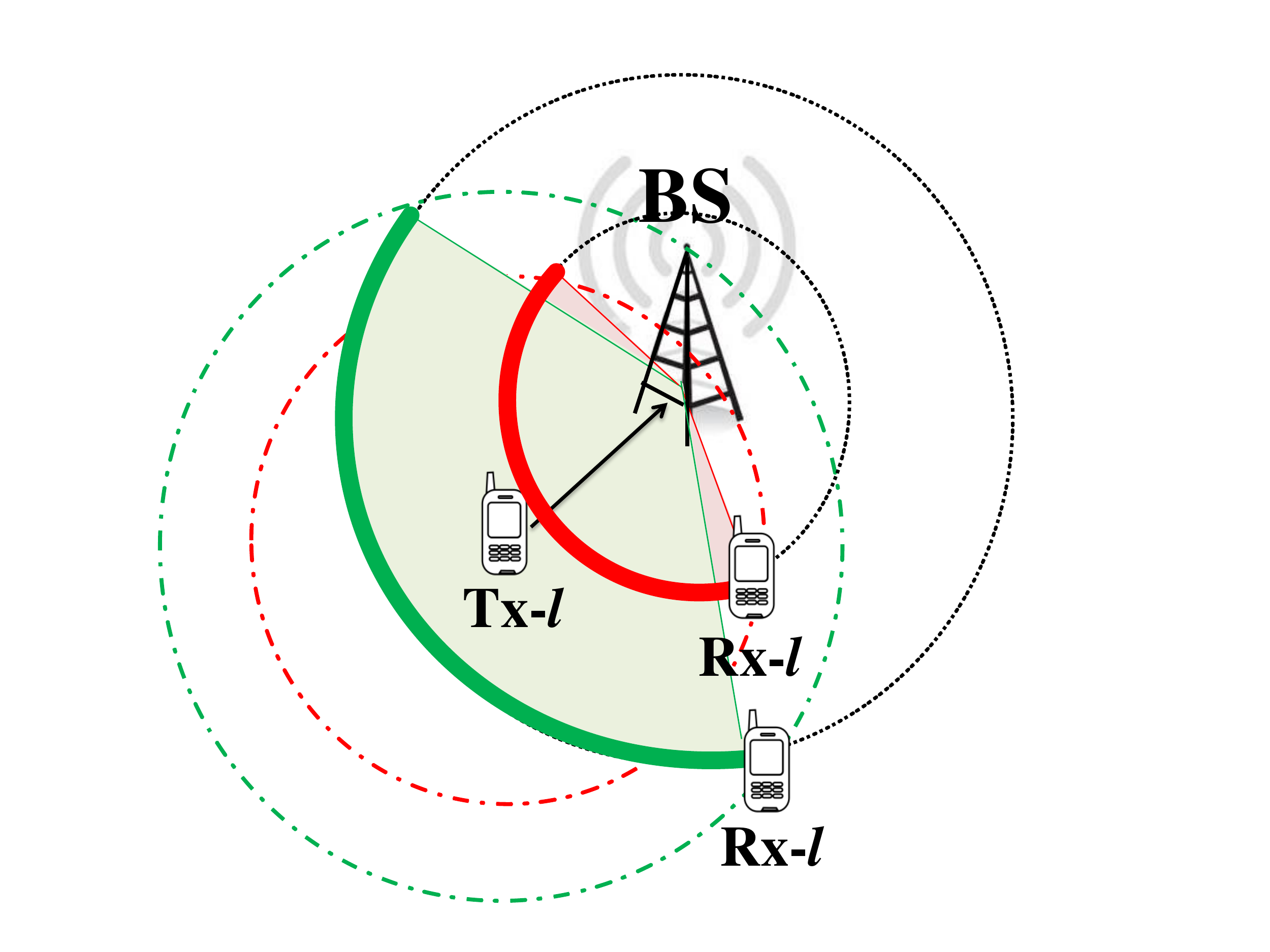}
\caption{Dashed circles centred at the position of Tx-$l$ have radius $\kappa(D_{0l})D_{l0}$ from \eqref{eqn:geometrical_interpretation}, and represent, for each of the two positions of Rx-$l$, the area within which D2D mode is more energy efficient than cellular mode for the user energy minimization.}\label{Res: SingleLink}
\end{figure}
%%%%%%

Similar calculations and arguments can be made for the system-energy (SE) objective. In this case,  $E_{ll}^{\rm CELL}(t_{\rm ul}^{\star})$ is given by (\ref{E_network}) and the expression for $\kappa(D_{0l})$ gets a bit more involved~\cite{PendaICC2015}. 
Fig.~\ref{Res: SingleLink SE} shows representative results for our simulation scenario. We observe that the D2D-optimal area is no longer circular and that the D2D mode is preferable in a large portion of the cell.

%%%%%
\begin{figure}[http]
     \begin{center}
        \subfigure[Tx-$l$ is 250m away from the BS.]{%
            \label{Res2_a}
            \includegraphics[width=50mm]{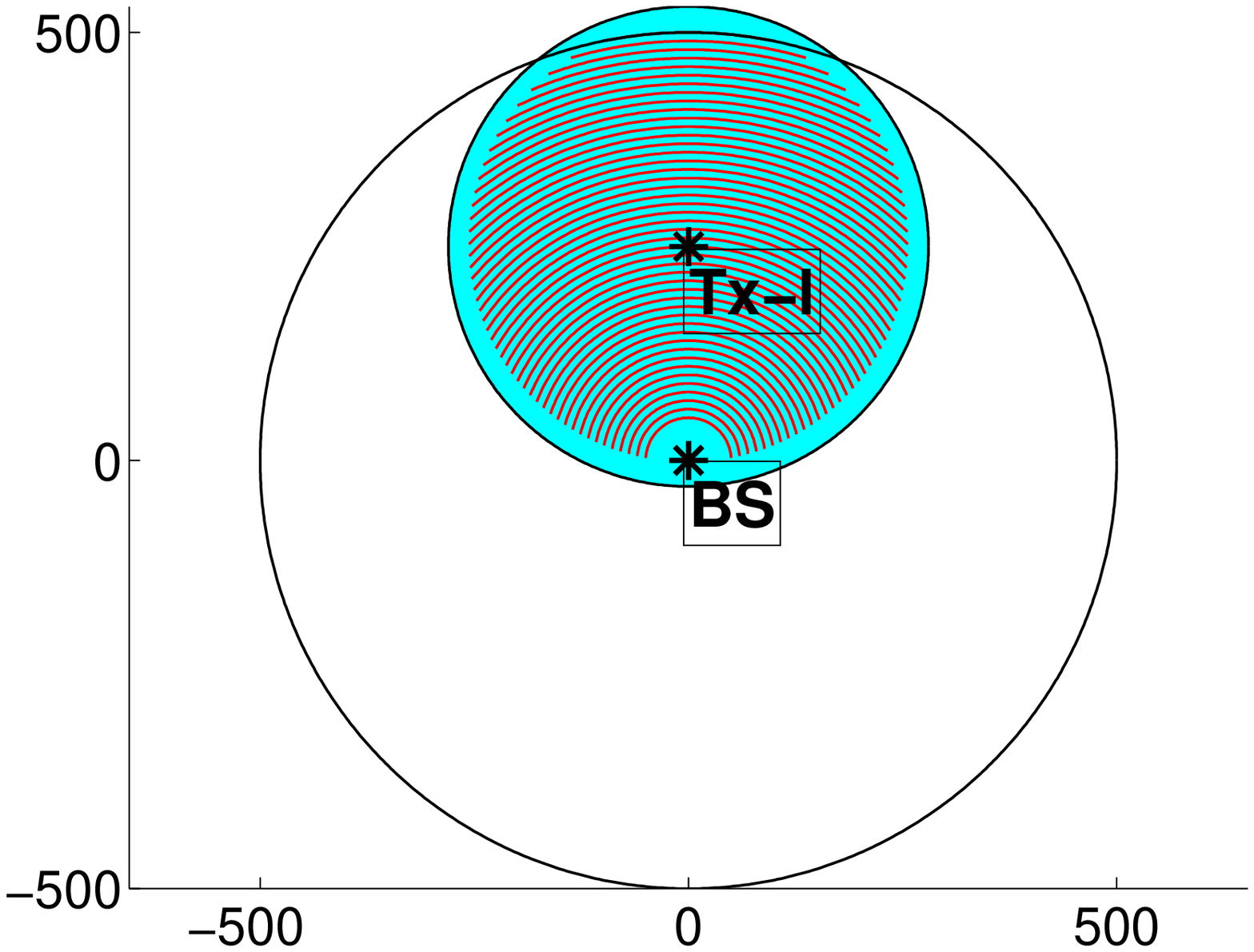}
        }
        \subfigure[Tx-$l$ is 450m away from the BS.]{%
            \label{Res2_b}
            \includegraphics[width=50mm]{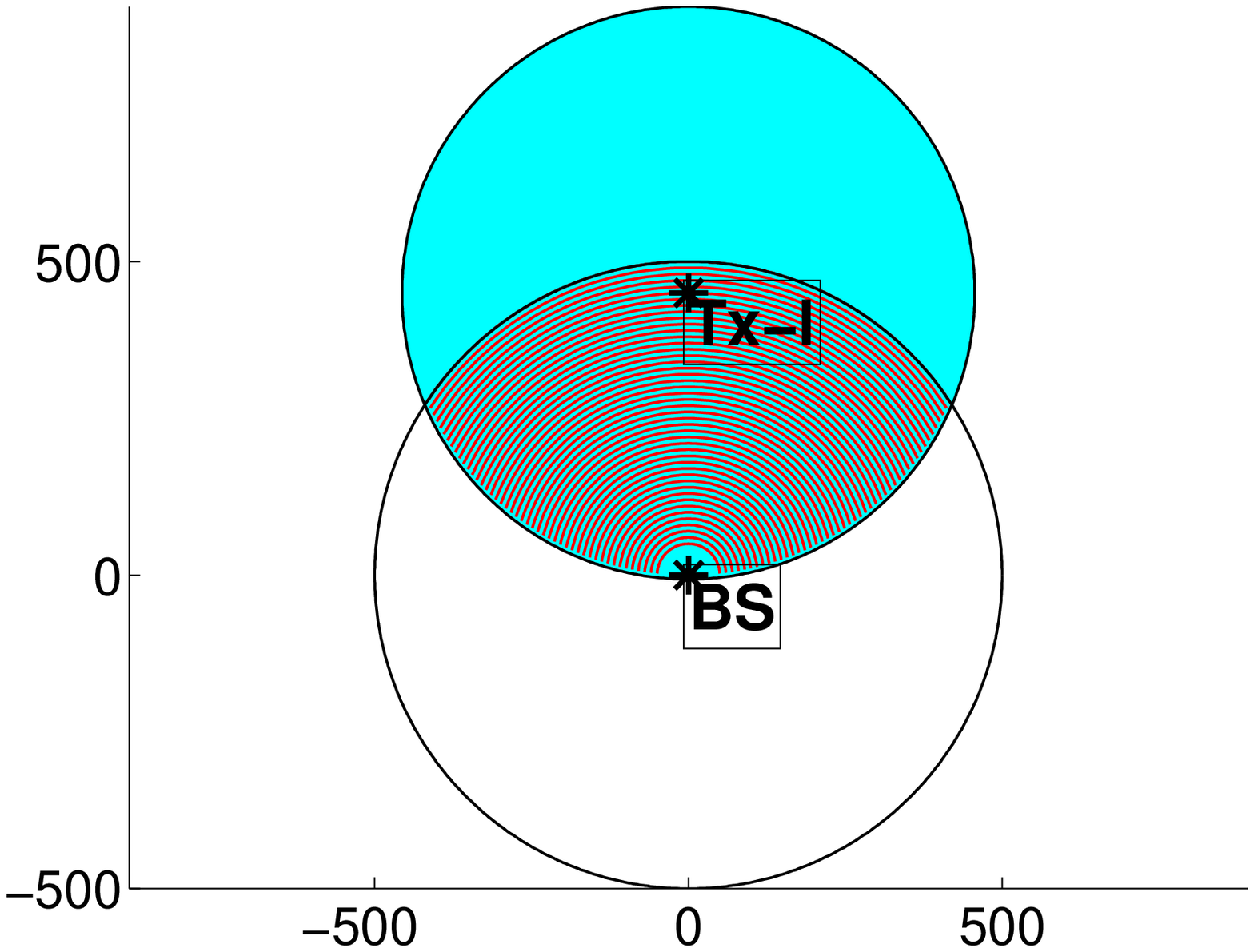}
        }
    \end{center}
    \caption{D2D optimality area when minimizing the mobile user energy consumption.  
    Red area represents the positions of Rx-$l$ for which D2D mode is more energy efficient than cellular mode, while light blue disk represents the area within which Tx-$l$ can fulfil the rate requirement transmitting in D2D mode with a feasible power level.}%
   \label{Res: SingleLink UE}
\end{figure}

%%%%%
\begin{figure}[http]
     \begin{center}
        \subfigure[Tx-$l$ is 250m away from the BS.]{%
            \label{Res3_a}
            \includegraphics[width=50mm]{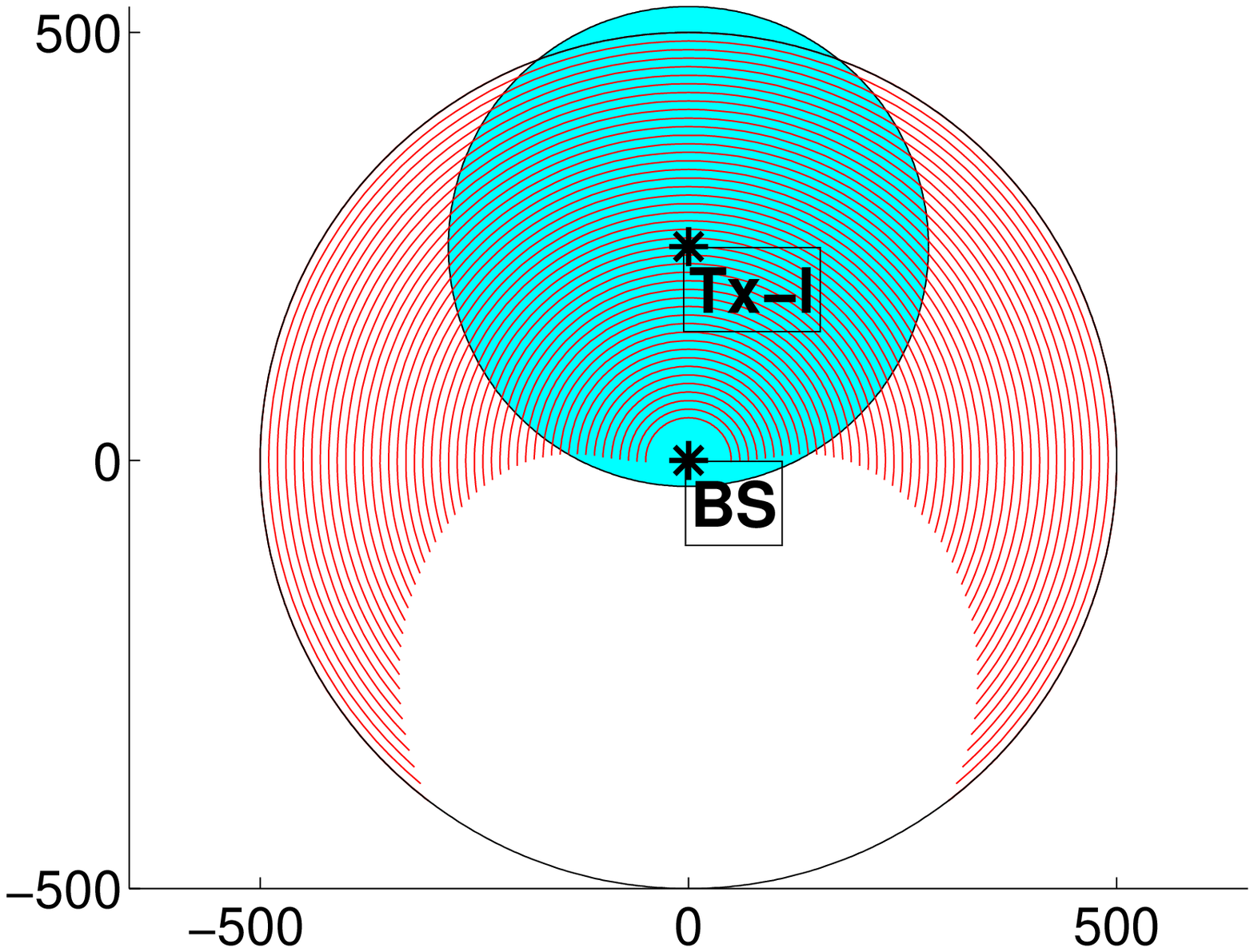}
        }
        \subfigure[Tx-$l$ is 450m aeay from the BS.]{%
            \label{Res3_b}
            \includegraphics[width=50mm]{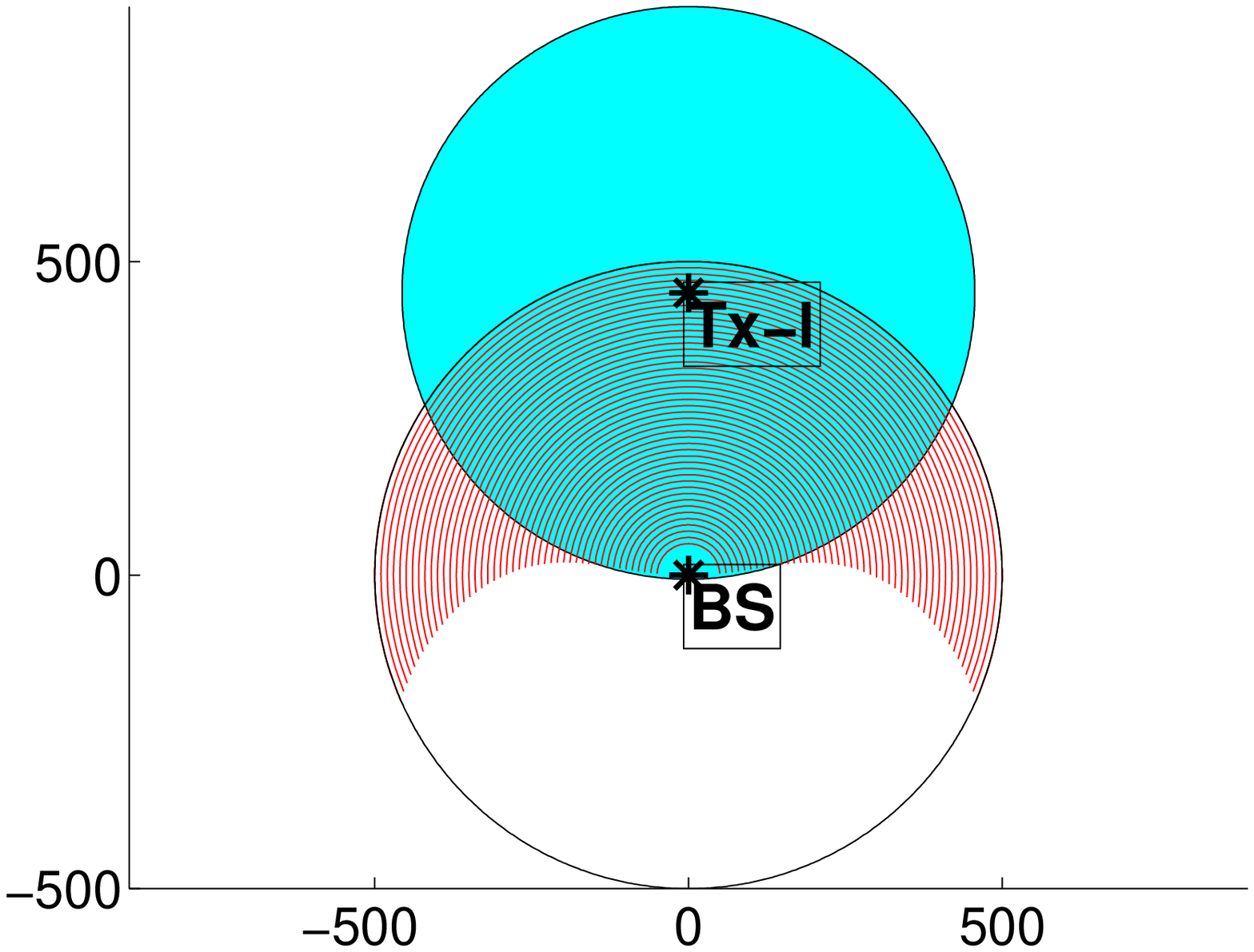}
        }
    \end{center}
    \caption{D2D optimality area when minimizing the system energy consumption. 
    Red area represents the positions of Rx-$l$ for which D2D mode is more energy efficient than cellular mode, while light blue disk represents the area within which Tx-$l$ can fulfil the rate requirement transmitting in D2D mode with a feasible power level.}%
   \label{Res: SingleLink SE}
\end{figure}
%%%%%

\subsection{Multiple link analysis}
The simulation experiments for the multi-link case are set up as follows. We generate random network topologies with given number of user pairs. Transmitters and receivers are randomly placed within the cell area. An example of network with 10 user pairs is given in Fig.~\ref{Res: Example}, where red squares represent the transmitters and green circles represent the receivers, transmitter and receiver forming a pair are labelled with the same number.
We assume large-scale path loss model, where gains are computed as described in Section~\ref{Res: sinlge link}.
Without loss of generality, we further assume the same max transmission power level for all mobile transmitters, and the same traffic requirement for all pairs, indicated with $b$.
To ensure Assumption~\ref{Ass1}, we set $b = \frac{r_{\rm ul}^{\max} r_{\rm dl}^{\max}}{r_{\rm ul}^{\max}+r_{\rm dl}^{\max}} T$, where $r_{\rm ul}^{\max}$ and $r_{\rm dl}^{\max} $ are the maximum achievable rate in UL and DL, respectively, when transmitter and receiver are both at the cell edge. 
For a given number of user pairs, we investigate 1000 random networks and present the averaged results.

%%%%%%
\begin{figure}[h]
\begin{center}
\includegraphics[scale=0.35]{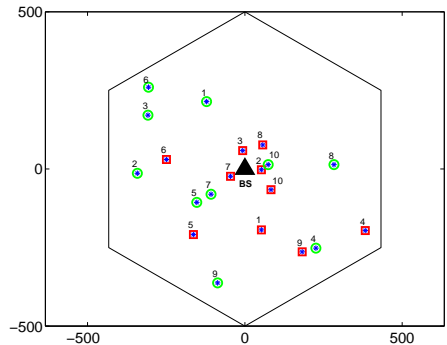}
\end{center}
\caption{Network with 10 user pairs randomly placed in a cell of radius 500m. The red squares represent the transmitters and the green circles represent the receivers. Transmitter and
receiver forming a pair are labelled with the same number.}\label{Res: Example}
\end{figure}
%%%%%%%

\subsubsection{Energy gain by enabling D2D communications in a fully orthogonal system}
To quantify the energy savings that can be obtained by exploiting direct communications, we compare the energy cost of the optimal FO-UE solution with the total user energy when all pairs are forced to communicate in cellular mode. For each of the 1000 random configurations used in our Monte Carlo study, we sort the links in order of increasing energy gain. Fig.~\ref{Res: Cell_vs_MS} shows the averaged results for networks with 10 and 30 user pairs, respectively.
%
%%%%%
\begin{figure}[h]
     \begin{center}
        \subfigure[Networks with 10 user pairs]{
            \includegraphics[scale=0.3]{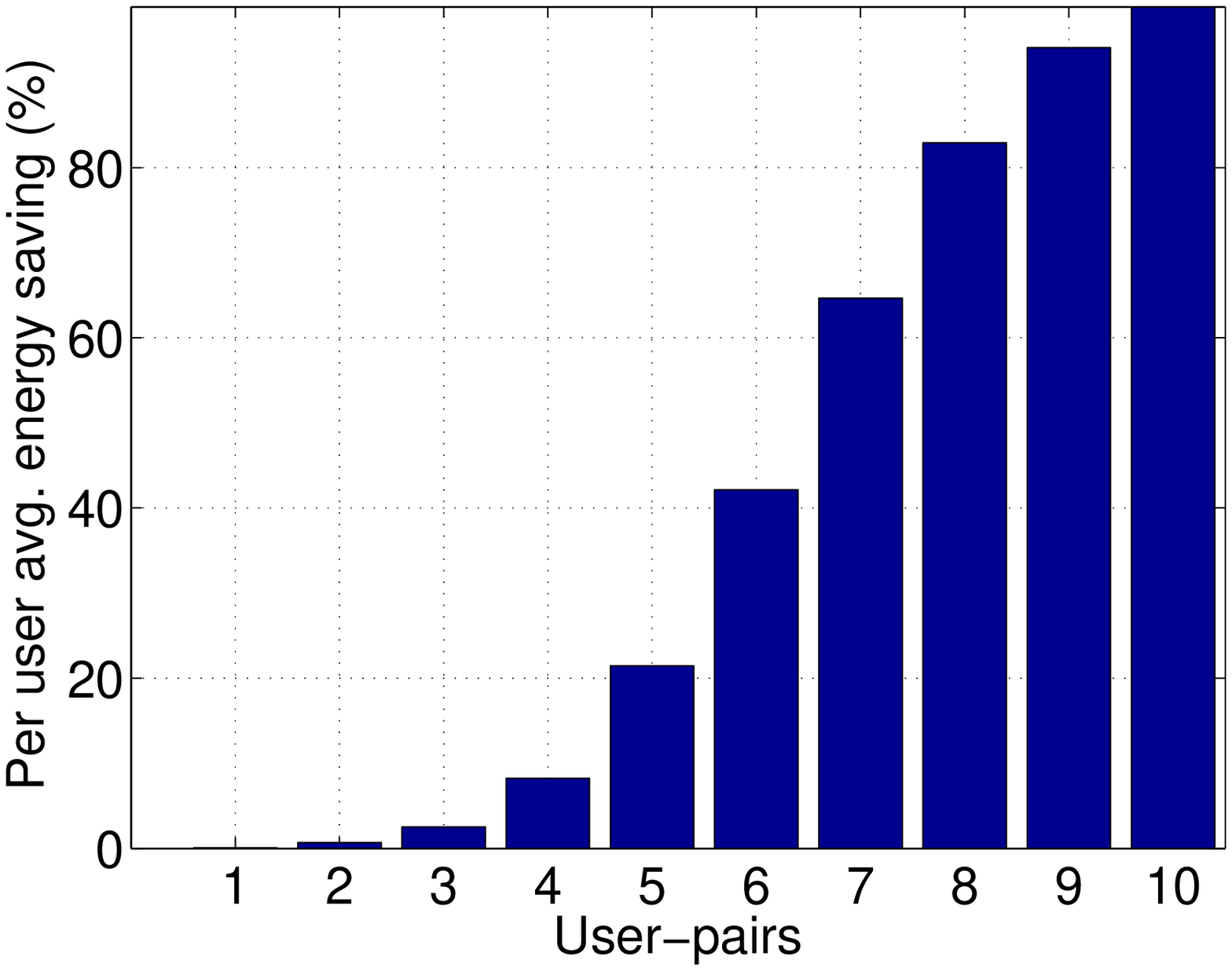}
        }
        \subfigure[Networks with 30 user pairs]{
           \includegraphics[scale=0.3]{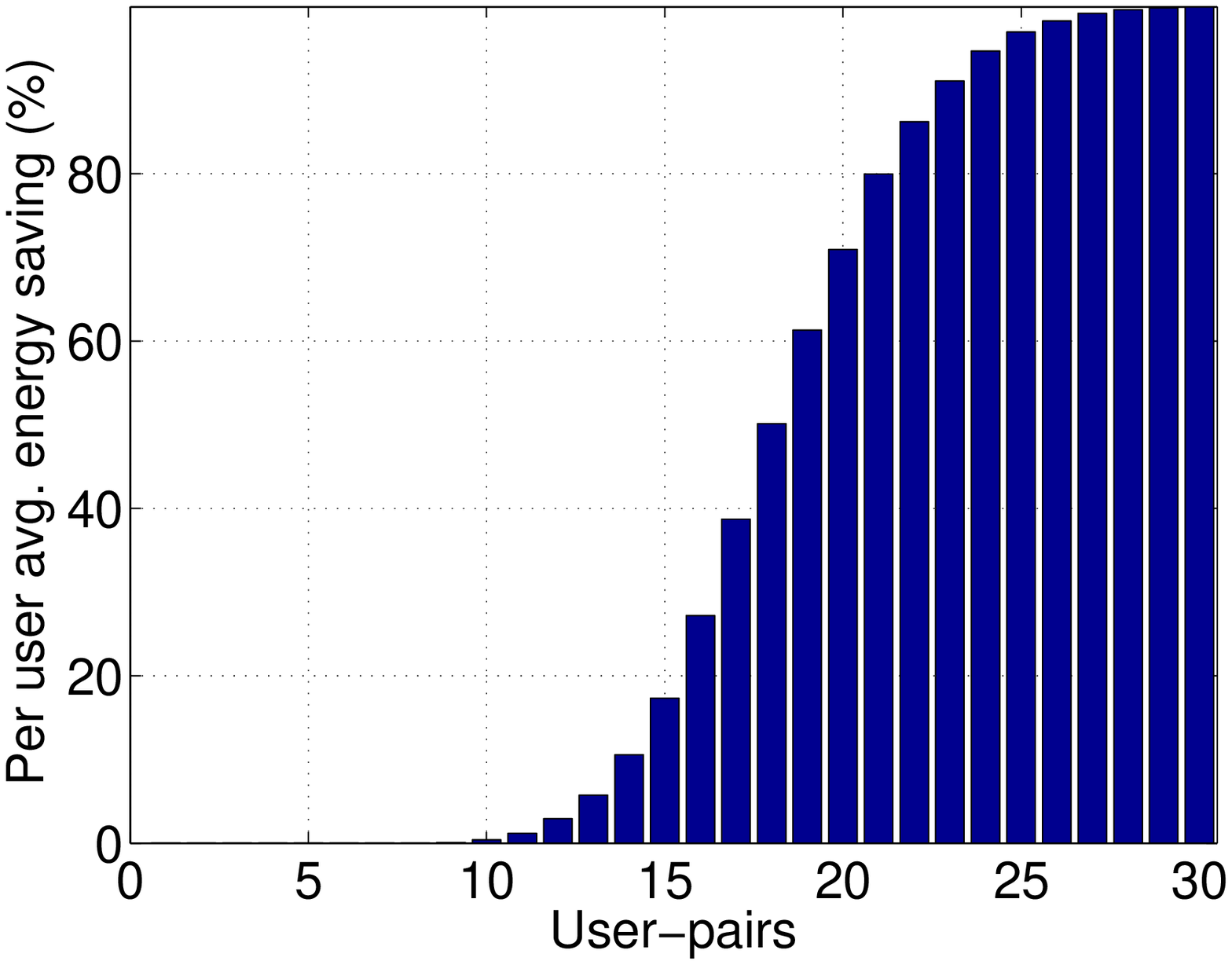}
        }
    \end{center}
\caption{Energy gain by enabling D2D communications in a fully orthogonal system. User pairs are sorted in increasing order of the energy gain they achieve by perfoeming mode selection, compared with traditional communication via the BS.}\label{Res: Cell_vs_MS}
\end{figure}
%%%%%

We observe that when D2D communication is enabled, each user either consumes the same amount of energy as in the traditional cellular mode or reduces its consumption (and sometimes significantly so). Indeed, Fig.~\ref{Res: Cell_vs_MS} demonstrates that, on average, half of the transmitters in the cell achieve an energy gain larger than $20\%$, and one third of the pairs have an energy gain larger than $60\%$. The average energy saving per user is around $40\%$.

It is important to notice that the energy gain is not only a consequence of the proximity of users, but also stems from the fact that D2D connections can use the full frame. Thus, it is the combination of D2D and flexible TDD technologies that contribute to the total energy savings.
 %%%%%%%%%%

%%%%%%%%%%%%%%%%%%%%%%%%%
\subsubsection{Performance evaluation of the B\&B algorithm for RS-UE}\label{Sec: results_B&B}
The difficulty in solving Problem~\eqref{PB_ML_I} lies mainly in the possible large search space of integer solutions. 
Table~\ref{Table2} shows the average number of mode selection vectors explored by different strategies before achieving the optimal solution. 
We consider the naive exhaustive enumeration algorithm, where we only eliminate infeasible solutions using Proposition~\ref{Prop: InfeasibleSubSet},  and the B\&B algorithm  described in Section~\ref{opt_sol}, both with the proposed branching rule and with a random selection of the branching variable. Averaged results show that the braching rules has a strong effect on the run-time.

\begin{table}[ht]
\centering
\begin{center}
\caption{Avg. number of explored integer solutions}
\label{Table2}
\end{center}
\begin{tabular}[h]{|l|c|c|}
\hline
\emph{\textbf{Algorithm}}     &   \emph{\textbf{10 user pairs}}   & \emph{\textbf{15 user pairs}} \\ 
\hline
Exhaustive enumeration  & 472.975 & $7.46\cdot10^3$\\
\hline
B\&B -  Random branching rule & 69.472 & 251.21\\
\hline
B\&B -  Proposed branching rule & 25.57 & 54.72 \\
\hline
\end{tabular}
\end{table}

%%%%%%%%%%%%%%%%%%%%%%%%%
\subsubsection{Performance evaluation of the heuristic mode selection algorithm for RS-UE}
In this section, we evaluate the performance of the heuristic mode selection policy for the RS-UE problem.  Fig.~\ref{Res: hist} shows the additional energy cost of the heuristic relative to the optimal solution computed using the B\&B solver. For networks with 10 user pairs,  the heuristic is within 10\% of the optimal solution for almost all network configurations (Fig.~\ref{Res: hist_a}).
For networks with 30 user pairs, the heuristic performs slightly worse. This performance degradation is due to the larger degree of freedom in packing D2D links on the same frequency channel. However, it still remains smaller than $10\%$ for most of the configurations. 

In both plots in Fig.~\ref{Res: hist}, there are a few rare network configurations where the heuristic performs much worse than optimal algorithm.
Fig.~\ref{Res: Example} is an example of such a scenario.
Under full orthogonality,  both pair 9 and pair 4 in Fig.~\ref{Res: Example} would be assigned to D2D mode.
When the heuristic initially attempts to assign these links to the same channel, it encounters an infeasible configuration due to the high interference that Rx-4  perceives from Tx-9. Therefore, only one of the two pairs can be assigned to D2D mode. The optimal decision is to let pair $4$ in D2D mode, but during the first iterations of the  power control in the heuristic, the interference from Rx-9 on Tx-4 leads pair-4 to leave the shared channel and switch to cellular mode. 

%%%%%%
\begin{figure}[h]
     \begin{center}
        \subfigure[Networks with 10 user pairs.]{\label{Res: hist_a}
            \includegraphics[scale=0.3]{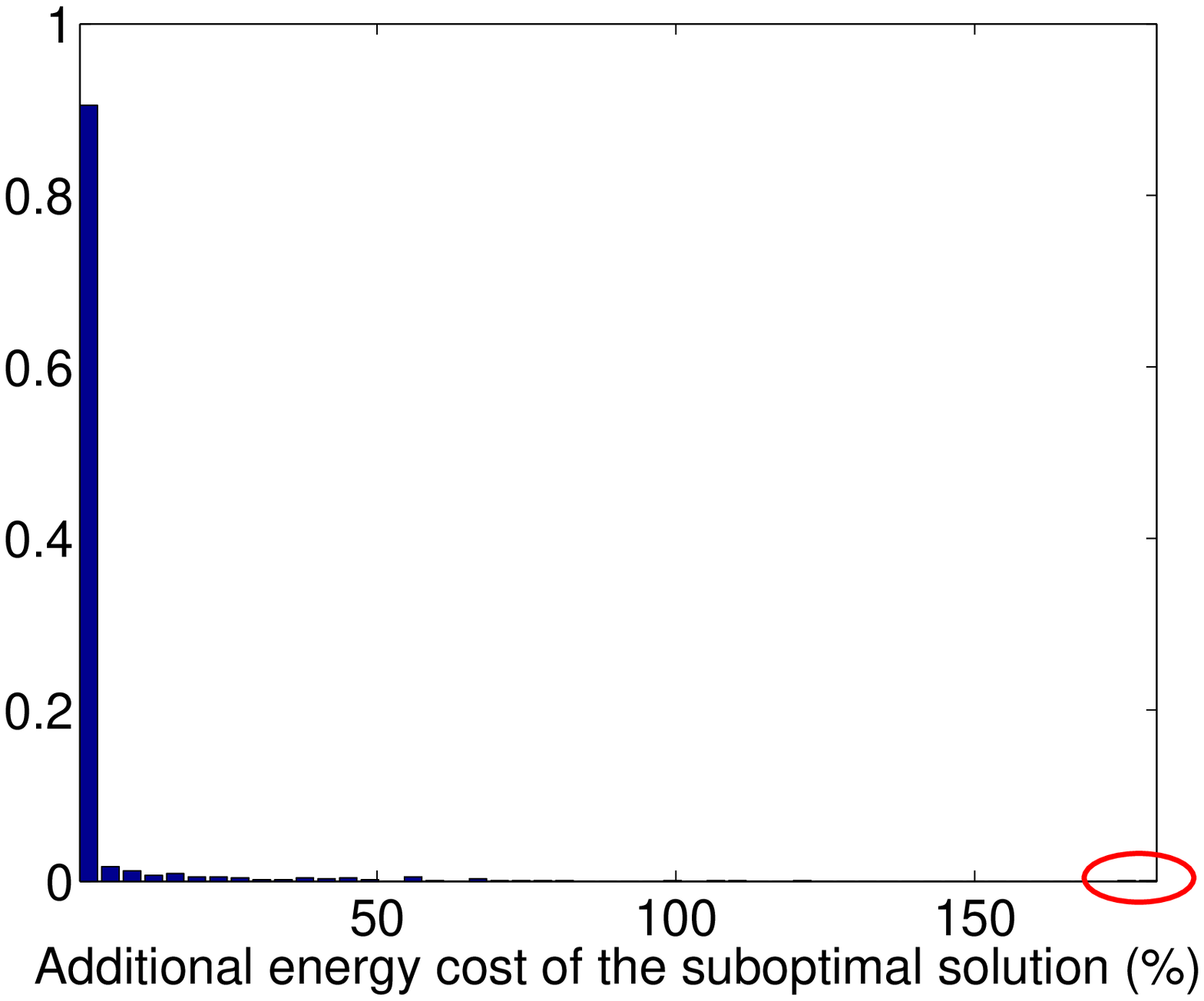}
        }
        \subfigure[Networks with 30 user pairs.]{\label{Res: hist_b}
           \includegraphics[scale=0.3]{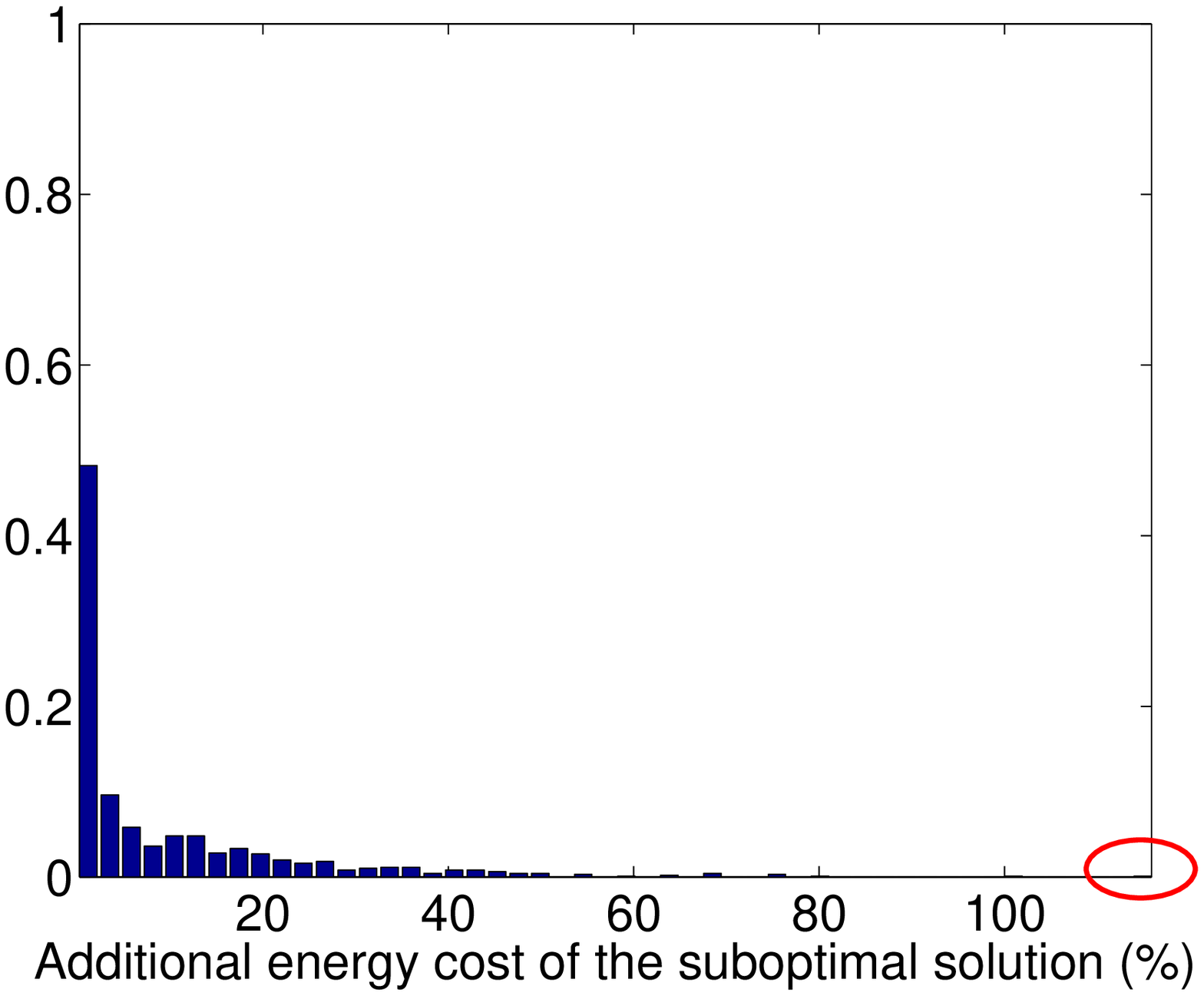}
        }
    \end{center}
\caption{Performance evaluation of the heuristic mode selection algorithm for RS-UE. Additional energy cost of the heuristic relative to the optimal solution achieved with B\&B .}\label{Res: hist}
\end{figure}
%%%%%%

%%%%%%%%%%%%%%%%%%%%%%%%%
\subsubsection{Performance comparison of different mode selection policies}
We have developed a number of optimal and heuristic mode selection policies. 
The aim of this section is to compare these different approaches and to bring additional insight in how they perform. 

Clearly, allowing for resource sharing reduces the number of orthogonal channels required, but increases the energy cost due to the interference among D2D links. Moreover, the heuristic allows to penalize the use of channel resources by tuning the threshold value $\theta$.
Fig.~\ref{Res: MS_comparison} shows the energy-channel performance for the different proposed algorithms in simulations of networks with 10 and 30 user pairs, respectively.
We note that  FO-UE represents the energy-optimal solution, using the same amount of frequency channels as the cellular mode. The RS scenario, in which all D2D pairs share the same frequency resource, uses fewer channels but a slightly higher energy cost due to the interference between D2D users.
The small increase in energy consumption can be understood by noting that the transmission powers assigned to D2D pairs are generally very small, partly because D2D pairs typically have high direct gains (since the transmitter and receiver often are in close proximity of each other), and partly since D2D connections can use the full frame. These results demonstrate that  D2D communications in cellular networks have the potential to improve both spectrum and energy efficiency over a traditional cellular solution. 

Finally, we evaluate the performance of the heuristic method for different values of the threshold $\theta$.
As expected, large values of $\theta$ decrease the number of channels used at the expense of a slightly increased energy consumption. For this reason, $\theta$ is an important design parameter for finding a suitable trade-off between energy consumption and channel use.

\begin{figure}[h]
     \begin{center}
        \subfigure[Networks with 10 user pairs.]{
            \includegraphics[scale=0.3]{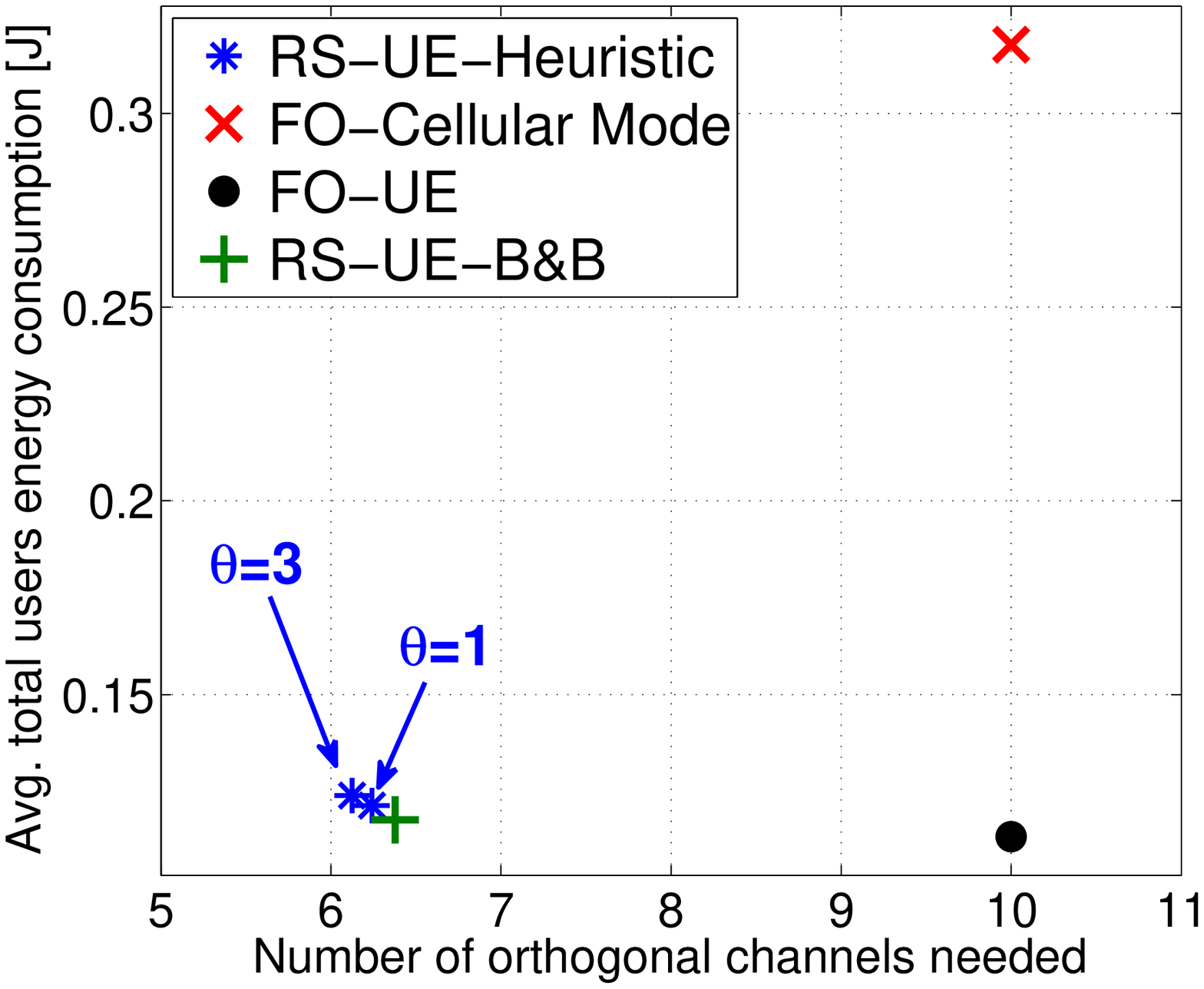}
        }
        \subfigure[Networks with 30 user pairs.]{
           \includegraphics[scale=0.3]{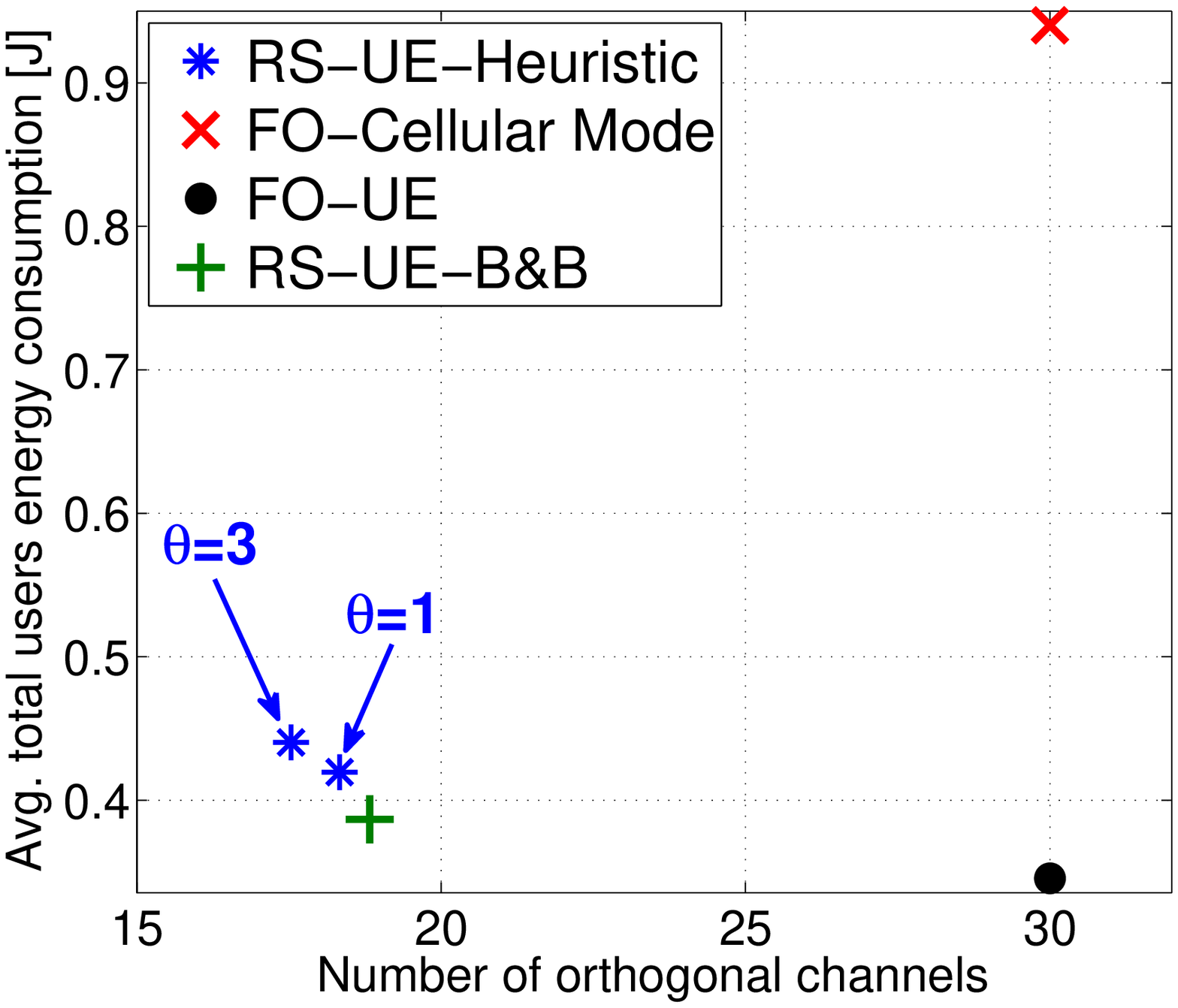}
        }
    \end{center}
\caption{Performance evaluation for different algorithms. Total energy consumption and number of orthogonal frequency channels needed to accomodate all the communication requests within the cell. Values are averaged over 1000 random configurations.} \label{Res: MS_comparison}
\end{figure}

%%%%%%%%%%%%%%%%%%%%%%%%%
%%%%%% CONSLUSION %%%%%
\section{Conclusions} \label{Sec: conclusion}
We investigated the problem of energy efficient mode selection and resource allocation for network-assisted D2D communications in a dynamic TDD system.
The problem has been analysed under two frequency channel allocation strategies (with and without interference among D2D pairs) and with two objectives (total user energy and total system energy). For each configuration we derived the optimal solution to the corresponding MINLP formulation. For the full orthogonality scenario, we demonstrated how the optimal solution could be obtained in polynomial (and sometimes even linear) time. When D2D pairs interfere with each other, on the other hand, we found it much harder to compute the optimal solution in an efficient manner. A customized branch-and-bound solver was therefore complemented by a more practical low-complexity heuristic. Through numerical simulations, we found that network-assisted D2D communications can yield significant energy savings, and that the heuristic algorithm could find near-optimal solutions while respecting practical implementation constraints.

\section*{Acknowledgment}
The authors thank Dr. Gabor Fodor, Dr. Themistoklis Charalambous and Dr. Euhanna Ghadimi,
whose comments helped to improve the presentation and the contents of the paper.

\ifCLASSOPTIONcaptionsoff
  \newpage
\fi

% -------------------------------------------------------------
% Bibliography
% -------------------------------------------------------------
\bibliographystyle{IEEEtran}
\bibliography{EnergyEfficientD2D}

\end{document}